%% file: main.tex
\documentclass[a4paper, USenglish, cleveref, autoref, thm-restate, authorcolumns]{lipics-v2019}
\usepackage[utf8]{inputenc}
\pdfoutput=1
\title{Approximate Search for Known Gene Clusters in New Genomes Using \PQT{s}}
\author{Galia R. Zimerman}{Ben Gurion University of the Negev, Israel}{zimgalia@gmail.com}{}{}
\author{Dina Svetlitsky}{Ben Gurion University of the Negev, Israel}{dina.svetlitsky@gmail.com}{}{}
\author{Meirav Zehavi}{Ben Gurion University of the Negev, Israel$^1$}{meiravze@bgu.ac.il}{}{}
\author{Michal Ziv-Ukelson}{Ben Gurion University of the Negev, Israel\footnote{Corresponding authors.}}{michaluz@cs.bgu.ac.il}{}{}
\authorrunning{G.\,R. Zimerman, D. Svetlitsky, M. Zehavi and M. Ziv-Ukelson} 
\Copyright{Galia R. Zimerman, Dina Svetlitsky, Meirav Zehavi and Michal Ziv-Ukelson}
\ccsdesc[500]{Applied computing~Bioinformatics}
\keywords{PQ-Tree, Gene Cluster, Efflux Pump} 
\supplement{\url{github.com/GaliaZim/PQFinder}}
\funding{The research of G.Z. was partially supported by the Planning and Budgeting Committee of the Council for Higher Education in Israel. The research of G.Z. and M.Z. was partially supported by the Israel Science Foundation (grant no. 1176/18). The research of G.Z., D.S. and M.Z.U. was partially supported by the Israel Science (grant no. 939/18). } 
\acknowledgements{Many thanks to Lev Gourevitach for his excellent implementation of a \pqt{} builder.}
\nolinenumbers 
\hideLIPIcs  

\input{preamble.tex}
\input{fig-preamble.tex}
\begin{document}
\maketitle
\begin{abstract}
We define a new problem in comparative genomics, denoted \up{}, that takes as input a \pqt{} $T$ representing the known gene orders of a gene cluster of interest, a gene-to-gene substitution scoring function $h$, integer parameters $d_T$ and $d_S$, and a new genome $S$. The objective is to identify in $S$ approximate new instances of the gene cluster that could vary from the known gene orders by genome rearrangements that are constrained by $T$, by gene substitutions that are governed by $h$, and by gene deletions and insertions that are bounded from above by $d_T$ and $d_S$, respectively.
We prove that the \up{} problem is \NPH{} and propose a parameterized algorithm that solves the optimization variant of \up{} in $O^*(2^{\gamma})$ time, where $\gamma$ is the maximum degree of a node in $T$ and $O^*$ is used to hide factors polynomial in the input size.

The algorithm is implemented as a search tool, denoted \alg{}, and applied to search for instances of chromosomal gene clusters in plasmids, within a dataset of 1,487 prokaryotic genomes. We report on 29 chromosomal gene clusters that are rearranged in plasmids, where the rearrangements are guided by the corresponding \pqt{}. One of these results, coding for a heavy metal efflux pump, is further analysed to exemplify how \alg{} can be harnessed to reveal interesting new structural variants of known gene clusters. 

\noindent{\bf Availability}  The code for the tool as well as all the data needed to reconstruct the results are publicly available on GitHub (\url{github.com/GaliaZim/PQFinder}).
\end{abstract}

\input{introduction}
\input{Preliminaries.tex}
\input{algorithm/algorithm.tex}
\input{methods}
\input{Bio-results}
\input{conclusion}
\bibliography{refs}

\renewcommand{\thefigure}{S\arabic{figure}}
\setcounter{figure}{0}
\renewcommand{\thetable}{S\arabic{table}}
\setcounter{table}{0}
\newpage
\appendix
\input{NPH.tex}
\input{appendix}
\input{naive}
\input{algorithm/Q-node}
\input{del-penalty}
\input{algorithm/proof/proofs.tex}

\input{appendix-figures}
\input{appendix-tables}

\end{document}

%% file: preamble.tex
\usepackage{standalone}
\usepackage{amsfonts}
\usepackage{amsthm}
\usepackage{amsmath}
\usepackage{graphicx,tikz}
\usepackage[linesnumbered,noline,ruled]{algorithm2e}
\usepackage{pgf}
\usepackage{enumerate}
\usepackage{todonotes}
\usepackage[super]{nth}
\usepackage{subcaption}
\usepackage[capitalise]{cleveref}
\usepackage{diagbox}
\usepackage{placeins}

\newcommand{\pqt}{PQ-tree}
\newcommand{\PQT}{PQ-Tree}
\newcommand{\R}{\mathbb{R}}
\newcommand{\N}{\mathbb{N}}
\newcommand{\jisp}{\textsc{Job Interval Selection} problem}
\newcommand{\JISP}[1]{\textsc{JISP#1}}
\newcommand{\A}{\mathcal{A}}
\newcommand{\xC}[1]{x^{(C #1)}}
\newcommand{\x}[1]{x^{(#1)}}
\newcommand{\xis}[1]{x_{[i#1]}}
\newcommand{\M}{\mathcal{M}}
\newcommand{\Q}{\mathcal{Q}}
\newcommand{\PP}{\mathcal{P}}

\newcommand{\D}{\mathcal{D}}
\newcommand{\Mc}{\mathcal{D}_\leq(C,k_T,k_S)}
\newcommand{\Mi}{\mathcal{D}_\leq(\xis{},k_T,k_S)}
\newcommand{\Mxc}{\mathcal{D}(\xC{},k_T,k_S)}
\newcommand{\Mxi}{\mathcal{D}(\x{i},k_T,k_S)}
\newcommand{\Mx}{\mathcal{D}(x,k_T,k_S)}
\newcommand{\Dm}{\mathcal{D}_M(x_j,i,k_T,k_S)}
\newcommand{\otom}{one-to-one mapping}
\newcommand{\up}{\textsc{\PQT{} Search}}
\newcommand{\ifff}{if and only if}
\newcommand{\xlabel}[1]{$\mathsf{label}(#1)$}

\newcommand{\dels}[1]{\mathsf{del}_S(#1)}
\newcommand{\delt}[1]{\mathsf{del}_T(#1)}
\newcommand{\len}[1]{\mathsf{len}(#1)}
\newcommand{\xspan}[1]{\mathsf{span}(#1)}

\newcommand{\children}[1]{\mathsf{children}(#1)}
\newcommand{\leaves}[1]{\mathsf{leaves}(#1)}
\newcommand{\rmdel}[2]{\mathsf{removeDel}(#1, #2)}
\newcommand{\adddel}[2]{\mathsf{addDel}(#1,#2)}
\newcommand{\remove}[2]{\mathsf{remove}(#1,#2)}
\newcommand{\add}[2]{\mathsf{add}(#1,#2)}
\newcommand{\abs}{{\em ab addition}}
\newcommand{\E}{E_{I}}
\newcommand{\dt}{del_T}
\newcommand{\ds}{del_S}
\newcommand{\alg}{PQFinder}
\newcommand{\OO}{\mathcal{O}}
\DeclareMathOperator*{\argmax}{arg\max}
\newcommand{\FPT}{\textsf{FPT}}
\newcommand{\NPH}{\textsf{NP}-hard}
\newcommand{\NPC}{\textsf{NP}-complete}

\newcommand{\para}[1]{\noindent\textbf{#1}}

%% file: fig-preamble.tex
\usepackage{tikz, animate}
\definecolor{darkred}{rgb}{0.8,0.0,0.0}
\definecolor{applegreen}{rgb}{0.5,0.81,0.0}
\usepackage{forest}
\forestset{
  *|/.style={
    parent anchor=children,     
    for children={
      edge path'={
        (!u.parent anchor-|.child anchor) -- (.child anchor)
      }
    }
  },
  *^/.style={
    for children={
      edge path'={
        (!u.parent) -- (.child anchor)
      }
    }
  },
  *-/.style={
    parent anchor=children,
    for children={
      edge path'={
        (!u.parent anchor-|.child anchor) -- (.child anchor) [color=white]
      } 
    }
  },
  declare dimen register=normal width,
  normal width'=.5cm,
  declare dimen register=normal height,
  normal height=.3cm,
  every forest node/.style={
    draw,
    minimum width/.register=normal width,
    minimum height/.register=normal height,
    inner sep=+0pt,
    anchor=children,
  },
  rect/.style={
    every forest node,
    shape=rectangle,
    tempdima={\pgfkeysvalueof{/pgf/outer xsep}},
    minimum width/.process={
      Rw+d nOw2+d Rw+d w3+d  {normal width}{##1*#1} {#1+1}{s sep}{##2*##1}  {tempdima}{##1*#1} {##1+##2-##3}
    },
  },
}

%% file: introduction.tex
\section{Introduction}\label{sec:intro}

Recent advances in pyrosequencing techniques, combined with global efforts to study infectious diseases, yield huge and rapidly-growing databases of microbial genomes \cite{tatusova2014refseq,wattam2014patric}. This big new data statistically empowers genomic-context based approaches to functional analysis: the biological principle underlying such analysis is that groups of genes that appear together consistently across many genomes often code for proteins that interact with one another, suggesting a common functional association. Thus, if the functional association and annotation 
of the clustered genes is already known in one (or more) of the genomes, this information can be used to infer functional characterization of homologous genes that are clustered together in another genome.

Groups of genes that are co-locally conserved across many genomes are denoted {\em gene clusters}. The locations of the group of genes comprising a gene cluster in the distinct genomes are denoted {\em instances}. Gene clusters in prokaryotic genomes often correspond to (one or several) operons; those are neighbouring genes that constitute a single unit of transcription and translation. However, the order of the genes in the distinct instances of a gene cluster may not be the same. 

The discovery (i.e. data-mining) of conserved gene clusters in a given set of genomes is a well studied problem \cite{bocker2009computation, he2005identifying,winter2016finding}. However, with the rapid sequencing of prokaryotic genomes a new problem is inspired: Namely, given an already known gene cluster that was discovered and studied in one genomic dataset, to identify all the instances of the gene cluster in a given new genomic sequence. 

One exemplary application for this problem is the search for chromosomal gene clusters in plasmids. Plasmids are circular genetic elements that are harbored by prokaryotic cells where they replicate independently from the chromosome. They can be transferred horizontally and vertically, and are considered a major driving force in prokaryotic evolution, providing mutation supply and constructing new operons with novel functions \cite{norris2013plasmids}, for example antibiotic resistance \cite{he2016mechanisms}.
This motivates biologists to search for chromosomal gene clusters in plasmids, and to study structural variations between the instances of the found gene clusters across the two distinct replicons. 
However, in addition to the fact that plasmids evolve independently from chromosomes and in a more rapid pace \cite{eberhard1990evolution}, their sequencing, assembly and annotation involves a more noisy process \cite{orlek2017plasmid}. 

To accommodate all this, the proposed search approach should be an approximate one, 
sensitive enough to tolerate some amount of genome rearrangements: transpositions and inversions, missing and intruding genes, and classification of genes with similar function to distinct orthology groups due to sequence divergence or convergent evolution. Yet, for the sake of specificity and search efficiency, we consider confining the allowed variations by two types of biological knowledge:  (1) bounding the allowed rearrangement events considered by the search, based on some grammatical model trained specifically from the known gene orders of the gene cluster, and (2) governing the gene-to-gene substitutions considered by the search by combining sequence homology with functional-annotation based semantic similarity. 

{\bf (1) Bounding the allowed rearrangement events.} The \pqt{} \cite{booth1976testing} is a combinatorial data structure classically used to represent gene clusters \cite{bergeron2008formal}.
A \pqt{} of a gene cluster describes its hierarchical inner structure and the relations between instances of the cluster succinctly, aids in filtering meaningful from apparently meaningless clusters, and also gives a natural and meaningful way of visualizing complex clusters.  A \pqt{} is a rooted tree with three types of nodes: {\em P-nodes}, {\em Q-nodes} and leaves. The children of a P-node can appear in any order, while the children of a Q-node must appear in either left-to-right or right-to-left order. 
(In the special case when a node has exactly two children, it does not matter whether it is labeled as a P-node or a Q-node.)
Booth and Lueker \cite{booth1976testing}, who introduced this data structure, were interested in representing a set of permutations over a set $U$, i.e. every member of $U$ appears exactly once as a label of a leaf in the \pqt{.} We, on the other hand, allow each member of $U$ to appear as a label of a leaf in the tree any non-negative number of times. Therefore, we will henceforth use the term {\em string} rather than {\em permutation} when describing the gene orders derived from a given \pqt{}.

\begin{figure}[t]
	\centering
	\includegraphics[width=\linewidth]{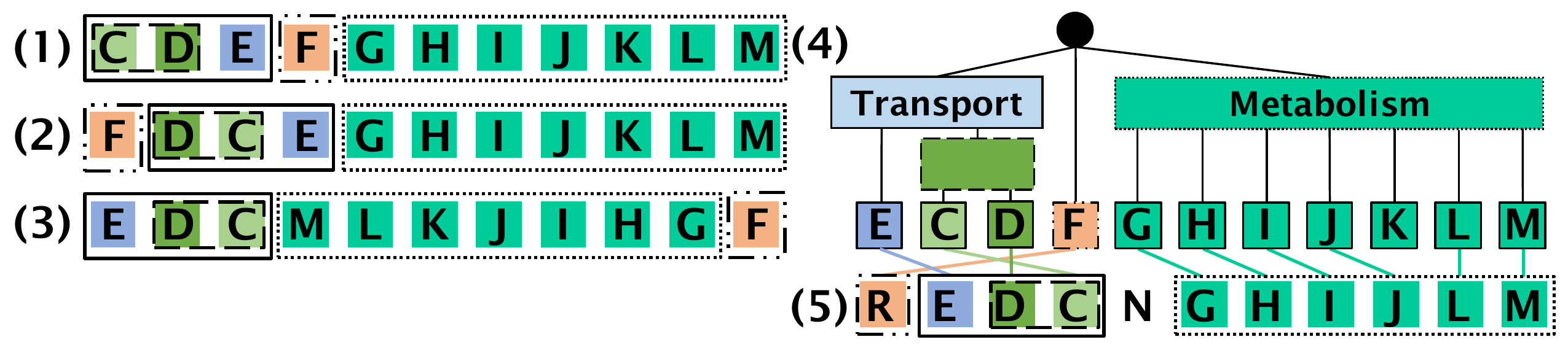}
	\caption{{\small A gene cluster containing most of the genes of the \textit{PhnCDEFGHIJKLMNOP} operon \cite{metcalf1993evidence} and the corresponding \pqt{.} The \textit{Phn} operon encodes proteins that utilize phosphonate as a nutritional source of phosphorus in prokaryotes. The genes \textit{PhnCDE} encode a phosphonate transporter, the genes \textit{PhnGHIJKLM} encode proteins responsible for the conversion of phosphonates to phosphate, and the gene \textit{PhnF} encodes a regulator. {\bf (1)-(3).} The three distinct gene orders found among 47 chromosomal instances of the $Phn$ gene cluster. 
	{\bf (4).} A \pqt{} representing the $Phn$ gene cluster, constructed from its three known gene orders shown in {1-3}. {\bf (5).} An example of a $Phn$ gene cluster instance identified by the \pqt{} shown in (4),  and the \otom{} between the leaves of the \pqt{} and the genes comprising the instance. The instance genes are rearranged differently from the gene orders shown in {1-3} and yet can be derived from the \pqt{}. In this mapping, gene $F$ is substituted by gene $R$, gene $N$ is an intruding gene (i.e., deleted from the instance string), and gene $I$ is a missing gene (i.e., deleted from the \pqt{}). 
}}
\label{fig:Phn}
\end{figure}

An example of a \pqt{} is given in \cref{fig:Phn}. It represents a $Phn$ gene cluster that encodes proteins that utilize phosphonate as a nutritional source of phosphorus in prokaryotes \cite{metcalf1993evidence}. 
The biological assumptions underlying the representation of gene clusters as \pqt{s} is that operons evolve via progressive merging of sub-operons, where the most basic units in this recursive operon assembly are colinearly conserved sub-operons \cite{fondi2009origin}. In the case where an operon is assembled from sub-operons that are colinearly dependent, the conserved gene order could correspond, e.g., to the order in which the transcripts of these genes interact in the metabolic pathway in which they are functionally associated \cite{wells2016operon}. Thus, transposition events shuffling the order of the genes within this sub-operon could reduce its fitness. On the other hand, inversion events, in which the genes participating in this sub-operon remain colinearly ordered are accepted. This case is represented in the \pqt{} by a Q-node (marked with a rectangle). In the case where an operon is assembled from sub-operons that are not colinearly co-dependent, convergent evolution could yield various orders of the assembled components \cite{fondi2009origin}. This case is represented in the \pqt{} by a P-node (marked with a circle). 
Learning the internal topology properties of a gene cluster from its corresponding gene orders and constructing a query \pqt{} accordingly, could empower the search to confine the allowed rearrangement operations so that colinear dependencies among genes and between sub-operons are preserved.

{\bf (2) Governing the gene-to-gene substitutions.} A prerequisite for gene cluster discovery is to determine how genes relate to each other across all the genomes in the dataset. In our experiment, genes are represented by their membership in Clusters of Orthologous Groups (COGs) \cite{tatusov2000cog}, where the sequence similarity of two genes belonging to the same COG serves as a proxy for homology.  
Despite low sequence similarity, genes belonging to two different COGs could have a similar function, which would be reflected in the functional description of the respective COGs. Using methods from natural language processing \cite{salton1975vector}, we compute for each pair of functional descriptions a score reflecting their semantic similarity. Combining sequence and functional similarity could increase the sensitivity of the search and promote the discovery of systems with related functions. 

\subparagraph{Our Contribution and Roadmap.}
In this paper we define a new problem in comparative genomics, denoted \up{} {\bf (in \cref{sec:prelim})}, that takes as input a \pqt{} $T$ (the query) representing the known gene orders of a  gene cluster of interest, a gene-to-gene substitution scoring function $h$, integer parameters $d_T$ and $d_S$, and a new genome $S$ (the target). The objective is to identify in $S$ a new approximate instance of the gene cluster that could vary from the known gene orders by genome rearrangements that are constrained by $T$, by gene substitutions that are governed by $h$, and by gene deletions and insertions that are bounded from above by $d_T$ and $d_S$, respectively. We prove that \up{} is \NPH{} {\bf (\cref{theorem:up-nph} in \cref{sec:NPh})}.

We define an optimization variant of \up{} and propose an algorithm {\bf (in \cref{sec:algorithm})} that solves it in $O(n \gamma {d_T}^2 {d_S}^2 (m_p \cdot 2^{\gamma} + m_q))$ time, where $n$ is the length of $S$, $m_p$ and $m_q$ denote the number of P-nodes and Q-nodes in $T$, respectively, and $\gamma$ denotes the maximum degree of a node in $T$.
In the same time and space complexities, we can also report all approximate instances of $T$ in $S$ and not only the optimal one. 

The algorithm is implemented as a search tool, denoted \alg{}. The code for the tool as well as all the data needed to reconstruct the results are publicly available on GitHub (\url{github.com/GaliaZim/PQFinder}). The tool is applied to search for instances of chromosomal gene clusters in plasmids, within a dataset of 1,487 prokaryotic genomes. In our preliminary results {\bf (given in \cref{sec:results})}, we report on 29 chromosomal gene clusters that are rearranged in plasmids, where the rearrangements are guided by the corresponding \pqt{}. One of these results, coding for a heavy metal efflux pump, is further analysed to exemplify how \alg{} can be harnessed to reveal interesting new structural variants of known gene clusters. 

\subparagraph{Previous Related Works.} Permutations on strings representing gene clusters have been studied earlier by 
\cite{bergeron2002algorithmic,eres2003combinatorial,heber2001algorithms,schmidt2004quadratic,uno2000fast}. 
\pqt{s} were previously applied in physical mapping \cite{alizadeh1995physical,christof1997branch}, as well as to other
comparative genomics problems \cite{berard2007perfect,bergeron2004reversal,landau2005gene}.

In Landau et al. \cite{landau2005gene} an algorithm was proposed for representation and detection of gene clusters in multiple genomes, using \pqt{s}: the proposed algorithm computes a \pqt{} of $k$ permutations of length $n$ in $O(kn)$ time, and it is proven that the computed \pqt{} is the one with a minimum number of possible rearrangements of its nodes while still representing all $k$ permutations.
In the same paper, the authors also present a general scheme to handle gene multiplicity and missing genes in permutations.
For every character that appears $a$ times in each of the $k$ strings, the time complexity for the construction of the \pqt{}, according to the scheme in that paper, is multiplied by an $O((a!)^k)$ factor.

Additional applications of PQ-trees to genomics were studied in \cite{adam2007common,bergeron2004reconstructing,parida2006using}, where PQ-trees were considered to represent and reconstruct ancestral genomes.

However, as far as we know, searching for approximate instances of a gene cluster that is represented as a \pqt{}, in a given new string, is a new computational problem.

%% file: Preliminaries.tex
\section{Preliminaries}\label{sec:prelim}
Let $\Pi$ be an NP-hard problem. In the framework of Parameterized Complexity, each instance of $\Pi$ is associated with a {\em parameter} $k$, and the goal is to confine the combinatorial explosion in the running time of an algorithm for $\Pi$ to depend only on $k$. 
Formally, $\Pi$ is {\em fixed-parameter tractable (\FPT{})} if any instance $(I, k)$ of $\Pi$ is solvable in time $f(k)\cdot |I|^{\OO(1)}$, where $f$ is an arbitrary computable function of $k$. 
Nowadays, Parameterized Complexity supplies a rich toolkit to design or refute the existence of \FPT{} algorithms \cite{DBLP:books/sp/CyganFKLMPPS15,DBLP:series/txcs/DowneyF13,fomin2019kernelization}.

\subparagraph{\PQT{:} Representing the Pattern.}
The possible reordering of the children nodes in a \pqt{} may create many equivalent \pqt{s}. Booth and Lueker \cite{booth1976testing} defined two \pqt{s} $T, \ T'$ as {\em equivalent} (denoted $T \equiv T'$) if one tree can be obtained by legally reordering the nodes of the other; namely, randomly permuting the children of a P-node, and reversing the children of a Q-node. To allow for deletions in the \pqt{s}, a generalization of their definition is given in \cref{def:quasi} below. Here, {\em smoothing} is a recursive process in which if by deleting leaves from a tree, $T$, some internal node $x$ of $T$ is left without children, then $x$ is also deleted, but its deletion is not counted (i.e. only leaf deletions are counted).
\begin{definition}[Quasi-Equivalence Between \PQT{s}] \label{def:quasi}
For any two \pqt{s}, $T$ and $T'$, the \pqt{} $T$ is {\em quasi-equivalent to $T'$ with a limit $d$}, denoted $T\succeq_d T'$, if $T'$ can be obtained from $T$ by (a) randomly permuting the children of some of the P-nodes of $T$, (b) reversing the children of some of the Q-nodes of $T$, and (c) deleting up to $d$ leaves from $T$ and applying the corresponding smoothing. (The order of the operations does not matter.)
\end{definition}
\Cref{equiv-pqts} shows two equivalent \pqt{s} (\cref{fig:pqt1}, \cref{fig:pqt2}) that are each quasi-equivalent with $d=1$ to the third \pqt{} (\cref{fig:pqt3}). 
The {\em frontier} of a \pqt{} $T$, denoted $F(T)$, is the sequence of labels on the leaves of $T$ read from left to right. For example, the frontier of the \pqt{} in \cref{fig:Phn} is $CDEFMLKJIHG$. 
It is interesting to consider the set of frontiers of all the equivalent \pqt{s}, defined in \cite{booth1976testing} as {\em consistent frontiers} and denoted by $C(T)=\{F(T'):T\equiv T'\}$. Intuitively, $C(T)$ is the set of all leaf label sequences defined by the \pqt{} structure and obtained by legally reordering its nodes. Here, we generalize the consistent frontiers definition to allow a bounded number of deletions from $T$, using quasi-equivalence.
\begin{definition}[$d$-Bounded Quasi-Consistent Frontiers]\label{def:C(T)}
$C_{d}(T)=\{F(T'):T\succeq_{d} T'\}$. 
\end{definition}
clearly $C_0(T)=C(T)$, and so in a setting where $d=0$ the latter notation is used.
For a node $x$ of a \pqt{} $T$, the subtree of $T$ rooted in $x$ is denoted by $T(x)$, the set of leaves in $T(x)$ is denoted by $\leaves{x}$, and the {\em span} of $x$ (denoted $\xspan{x}$) is defined as $|\leaves{x}|$.

\subparagraph{\PQT{} Search and Related Terminology.} \label{sec:up}
An instance of the \up{} problem is a tuple $(T, S, h, d_T, d_S)$, where $T$ is a \pqt{} with $m$ leaves, $m_p$ P-nodes, $m_q$ Q-nodes and every leaf $x$ in $T$ has a label \xlabel{x}$\in\Sigma_T$; $S=\sigma_1 \dotso \sigma_n \in {\Sigma_S^n}$ is a string of length $n$ representing the input genome; $d_T\in\mathbb{N}$ specifies the number of allowed deletions from $T$; $d_S\in\mathbb{N}$ specifies the number of allowed deletions from $S$; and $h$ is a {\em boolean substitution function}, describing the possible substitutions between the leaf labels of $T$ and the characters of the given string, $S$. 
Formally, $h$ is a function that receives a pair $(\sigma_t, \sigma_s)$, where $\sigma_t \in \Sigma_T$ is one of the labels on the leaves of $T$, and $\sigma_s \in \Sigma_S$ is one of the characters of the given string, $S$, and returns $True$ if $\sigma_t$ can be replaced with $\sigma_s$, and $False$, otherwise.
Considering the biological problem at hand, $\Sigma_T$ and $\Sigma_S$ are both sets of genes.
For $1\leq i\leq j \leq n$, $S'=S[i:j] = \sigma_i...\sigma_j$ is a substring of $S$ beginning
at index $i$ and ending at index $j$. The substring $S'$ is a {\em prefix} of $S$ if $S'=S[1:j]$ and it is a {\em suffix} of $S$ if $S'=S[i:n]$. In addition,  we denote $\sigma_i$, the $i^{\textrm{th}}$ character of $S$, by $S[i]$.
 
The objective of \up{} is to find a \otom{} $\M{}$ between the leaves of $T$ and the characters of a substring $S'$ of $S$, that comprises a set of pairs each having one of three forms: the substitution form, $(x,\sigma_s(\ell))$, where $x$ is a leaf in $T$, $\sigma_s \in \Sigma_S$, $h(\mathsf{label}(x),\sigma_s) = True$ and $\ell \in \{1,\dotso{,}n\}$ is the index of the occurrence of $\sigma_s$ in $S$ that is mapped to the leaf $x$; the character deletion form, $(\varepsilon, \sigma_s(\ell))$, which marks the deletion of the character $\sigma_s\in \Sigma_S$ from the index $\ell$ of $S$; the leaf deletion form, $(x, \varepsilon)$, which marks the deletion of $x$, a leaf node of $T$.

To account for the number of deletions of characters of $S'$ and leaves of $T$ in $\M{}$, the number of pairs in $\M{}$ of the form $(\varepsilon,\sigma)$ are marked by $\dels{\M{}}$ and the number of pairs in $\M{}$ of the form $(x,\varepsilon)$ are marked by $\delt{\M{}}$.
Applying the substitutions defined in $\M{}$ to $S'$ resulting in the string $S_{\M{}}$ is the process in which for every $(x,\sigma_s(\ell))\in \M{}$, the character $\sigma_s$ at index $\ell$ of $S$ is deleted if $x=\varepsilon$, and otherwise substituted by $x$. This process is demonstrated in \cref{fig:derivation-string}.
We say that $S'$ is {\em derived} from $T$ {\em under} $\M{}$ with $d_T$ deletions from the tree and $d_S$ deletions from the string, if $d_T=\delt{\M{}}$, $d_S=\dels{\M{}}$ and $S_{\M{}}\in C_{d_T}(T)$. Thus, by definition, there is a \pqt{} $T'$ such that $F(T')=S_{\M{}}$ and $T \succeq_{d_T} T'$. Note that the deletions of the nodes in $T$ to obtain the nodes in $T'$ are determined by $\M{}$. The conversion of $T$ to $T'$ as defined by the derivation is illustrated in \cref{fig:derivation-tree}.
The set of permutations and node deletions performed to obtain $T'$ from $T$ together with the substitutions and deletions from $S'$ specified by $\M{}$ is named the {\em derivation $\mu$} of $T$ to $S'$. We also say that $\M{}$ {\em yields} the derivation $\mu$.

For a derivation $\mu$ of $T$ to $S'=S[s:e]$, we give the following terms and notations (illustrated in \cref{fig:derivation}). The root of $T$ is {\em the node that $\mu$ derives} or {\em the root of the derivation} and it is denoted by $\mu.v$. For abbreviation, we say that $\mu$ {\em is a derivation of $\mu.v$}. The substring $S'$ is {\em the string that $\mu$ derives}. We name $s$ and $e$ the start and end points of the derivation and denote them by $\mu.s$ and $\mu.e$, respectively. The \otom{} that yields $\mu$ is denoted by $\mu.o$. The number of deletions from the tree is denoted by $\mu.\dt$. The number of deletions from the string is denoted by $\mu.\ds$.
In addition, if $x$ is a leaf node in $T$ and $(x,\sigma_s(\ell))\in \mu.o$, then $x$ is {\em mapped to} $S[\ell]$ {\em under} $\mu$. The character $S[\ell]$ is said to be {\em deleted under} $\mu$ if $(\varepsilon,\sigma_s(\ell))\in \mu.o$. If $x \in T(\mu.v)$ is a leaf for which $(x,\varepsilon)\in \mu.o$, then $x$ is {\em deleted under} $\mu$. For an internal node of $T$, $x$, if every leaf in $T(x)$ is deleted under $\mu$, then $x$ is {\em deleted under} $\mu$, and otherwise $x$ is {\em kept under} $\mu$.

We define two versions of the \up{} problem: a decision version (\cref{def:decision-up}) and an optimisation version (\cref{def:opti-up}).

\begin{definition}[Decision \PQT{} Search]\label{def:decision-up}
Given a string $S$ of length $n$, a \pqt{} $T$ with $m$ leaves, deletion limits $d_T,d_S \in \N{}$, and a boolean substitution function $h$ between $\Sigma_S$ and $\Sigma_T$, decide if there is a \otom{} $\M{}$ that yields a derivation of $T$ to a substring $S'$ of $S$ with up to $d_T$ and up to $d_S$ deletions from $T$ and $S'$, respectively.
\end{definition}

To define an optimization version of the \up{} problem it is necessary to have a score for every possible substitution between the characters in $\Sigma_T$ and the characters in $\Sigma_S$. Hence, for this problem variant assume that $h$ is a {\em substitution scoring function}, that is, $h(\sigma_t,\sigma_s)$ for $\sigma_t\in\Sigma_T, \sigma_s\in\Sigma_S$ is the score for substituting $\sigma_s$ by $\sigma_t$ in the derivation, and if $\sigma_t$ cannot be substituted by $\sigma_s$, $h(\sigma_t,\sigma_s) = -\infty$.
In addition, we need a cost function, denoted by $\delta$, for the deletion of a character of $S$ and for the deletion of a leaf of $T$ according to the label of the leaf.
The score of a derivation $\mu$, denoted by $\mu.score$, is the sum of scores of all operations (deletions from the tree, deletions from the string and substitutions) in $\mu$. 
Now, instead of deciding whether there is a \otom{} that yields a derivation of $T$ to a substring of $S$, we can search for the \otom{} that yields the best derivation (if there exists such a derivation), i.e. a \otom{} for which $\mu.score$ is the highest.

\begin{definition}[Optimization \PQT{} Search]\label{def:opti-up}
Given a string of length $n$, $S$, a \pqt{} with $m$ leaves, $T$, deletion limits $d_T,d_S\in \N{}$, a substitution scoring function between $\Sigma_S$ and $\Sigma_T$, $h$, and a deletion cost function, $\delta$, return the \otom{}, $\M{}$, that yields the highest scoring derivation of $T$ to a substring $S'$ of $S$ with up to $d_T$ deletions from $T$ and up to $d_S$ deletions from $S'$ (if such a mapping exists).
\end{definition}

%% file: algorithm/algorithm.tex
\input{algorithm/main-algorithm}
\input{algorithm/P-node}
\input{algorithm/time-complex}

%% file: algorithm/main-algorithm.tex
\section{A Parameterized Algorithm} \label{sec:algorithm}
In this section we develop a dynamic programming (DP) algorithm to solve the optimization variant of \up{} (\cref{def:opti-up}). 
Our algorithm receives as input an instance of \up{} $(T, S, h, d_T, d_S)$, where $h$ is a substitution scoring function as defined in \cref{sec:up}.
Our default assumption is that deletions are not penalized, and therefore $\delta$ is not given as input. The case where deletions are penalized is described in \cref{sec:del-penalty}.
The output of the algorithm is a \otom{}, $\M{}$, that yields the best (highest scoring) derivation of $T$ to a substring of $S$ with up to $d_T$ deletions from $T$ and up to $d_S$ deletions from the substring, and the score of that derivation. With a minor modification, the output can be extended to include a \otom{} for every substring of $S$ and the derivations that they yield.

\subparagraph{Brief Overview.}
On a high level, our algorithm consists of three components: the main algorithm, and two other algorithms that are used as procedures by the main algorithm. Apart from an initialization phase, the crux of the main algorithm is a loop that traverses the given \pqt{,} $T$. For each internal node $x$, it calls one of the two other algorithms: P-mapping (given in \cref{subsec:p-mapping-alg}) and Q-mapping (given in \cref{sec:q-node-mapping}). These algorithms find and return the best derivations from the subtree of $T$ rooted in $x$, $T(x)$, to substrings of $S$, based on the type of $x$ (P-node or Q-node). Then, the scores of the derivations are stored in the DP table.

We now give a brief informal description of the main ideas behind our P-mapping and Q-mapping algorithms.
Our P-mapping algorithm is inspired by an algorithm described by Bevern et al. \cite{vanBevern2015} to solve the \jisp{}. Our problem differs from theirs mainly in its control of deletions.
Intuitively, in the P-mapping algorithm we consider the task at hand as a packing problem, where every child of $x$ is a set of intervals, each corresponding to a different substring. The objective is to pack non-overlapping intervals such that for every child of $x$ at most one interval is packed.
Then, the algorithm greedily selects a child $x'$ of $x$ and decides either to pack one of its intervals (and which one) or to pack none (in which case $x'$ is deleted).
Our Q-mapping algorithm is similar to the P-mapping algorithm, but simpler. It can be considered as an interval packing algorithm as well, however, this algorithm packs the children of $x$ in a specific order.

In the following sections, we describe the main algorithm, the P-mapping algorithm, and afterwards analyse the time complexity. The Q-mapping algorithm, which is also used as a procedure in the main algorithm, is described in \cref{sec:q-node-mapping}.

\subsection{The Main Algorithm}\label{sec:general-algorithm}
We now delve into more technical details. The algorithm (whose pseudocode is given in \cref{alg:up} in \cref{sec:proofs}) constructs a $4$-dimensional DP table $\A{}$ of size $m' \times n \times d_T+1 \times d_S+1$. The purpose of an entry of the DP table, $\A[j,i,k_T,k_S]$, is to hold the highest score of a derivation of the subtree $T(x_j)$ to a substring $S'$ of $S$ starting at index $i$ with $k_T$ deletions from $T(x_j)$ and $k_S$ deletions from $S'$. If no such derivation exists, $\A[j,i,k_T,k_S] = -\infty$. Addressing $\A{}$ with some of its indices given as dots, e.g. $\A[j, i,\cdot,\cdot]$, refers to the subtable of $\A{}$ that is comprised of all entries of $\A{}$ whose first two indices are $j$ and $i$.
Some entries of the DP table define illegal derivations, namely, derivations for which the number of deletions are inconsistent with the start index, $i$, the derived node and $S$. These entries are called {\em invalid entries} and their value is defined as $-\infty$ throughout the algorithm. A more detailed description of the invalid entries is given in \cref{sec:q-node-mapping}.

The main algorithm first initializes the entries of $\A{}$ that are meant to hold scores of derivations of the leaves of $T$ to every possible substring of $S$ using the following rule. For every $0\leq k_S\leq d_S$ and every $x_j\in \leaves{root}$, do: 
\begin{enumerate}
    \item $\A[j,i,1,k_S] = 0$
    \item $\A[j,i,0,k_S] = \displaystyle{\max_{\substack{i'=i,...,i+k_S}}}h(j,S[i'])$
\end{enumerate}
Afterwards, all other entries of $\A{}$ are filled as follows. Go over the internal nodes of $T$ in postorder. For every internal node, $x$, go in ascending order over every index, $i$, that can be a start index for the substring of $S$ derived from $T(x)$ (the possible values of $i$ are explained in the next paragraph). 
For every $x$ and $i$, use the algorithm for Q-mapping or P-mapping according to the type of $x$. 
Both algorithms receive the same input: a substring $S'$ of $S$, the node $x$, its children $x_{1},\dots,x_{\gamma}$, the collection of possible derivations of the children (denoted by $\D{}$), which have already been computed and stored in $\A{}$ (as will be explained ahead) and the deletion arguments $d_T,d_S$. 
Intuitively, the substring $S'$ is the longest substring of $S$ starting at index $i$ that can be derived from $T(x)$ given $d_T$ and $d_S$.
After being called, both algorithms return a set of derivations of $T(x)$ to a prefix of $S'=S[i:e]$ and their scores.
The set holds the highest scoring derivation for every $E(x_j,i,d_T,0) \leq e \leq E(x_j,i,0,d_S)$ and for every legal deletion combination $0\leq k_T\leq d_T$, $0\leq k_S \leq d_S$.

We now explain the possible values of $i$ and the definition of $S'$ more formally. To this end, note that given the node $x$ and some numbers of deletions $k_T$ and $k_S$, the length of the derived substring is $L(x,k_T,k_S) \doteq \xspan{x} - k_T + k_S$ (see \cref{par:len-derived-string}). 
Thus, on the one hand, a substring of maximum length is obtained when there are no deletions from the tree and $d_S$ deletions from the string.
Hence, $S'=S[i:E(x,i,0,d_S)]$ where $E(x,i,k_T,k_S)$ is the function for the calculation of the end point of a derivation, defined as $E(x,i,k_T,k_S) \doteq i - 1 + L(x,k_T,k_S)$.
On the other hand, a shortest substring is obtained when there are $d_T$ deletions from the tree and none from the string. Then, the length of the substring is $L(x,d_T,0) = \xspan{x}-d_T$. Hence, the index $i$ runs between $1$ and $n-(\xspan{x}-d_T)+1$.

We now turn to address the aforementioned input collection $\D{}$ in more detail. Formally, it contains the best scoring derivations of every child $x_{j}$ of $x$ to every substring of $S'$ with up to $d_T$ and $d_S$ deletions from the tree and string, respectively. It is produced from the entries $\A[j, i', k_T, k_S]$ (where each entry gives one derivation) for all $k_T$ and $k_S$, and all $i'$ between $i$ and the end index of $S'$, i.e. $i\leq i' \leq E(x_j,i,0,d_S)$.
For the efficiency of the Q-mapping and P-mapping algorithms, the derivations in $\D{}$ are arranged in descending order with respect to their end point ($\mu.e$). This does not increase the time complexity of the algorithm, as this ordering is received by previous calls to the Q-mapping and P-mapping algorithms.

In the final stage of the main algorithm, when the DP table is full, the score of a best derivation is the maximum of $\{\A[m',i,k_T,k_S] : k_T\leq d_T$, $k_S \leq d_S$, $1\leq i\leq n-(\xspan{root}-k_T)+1\}$ (remember that $x_{m'}$ is the root of $T$). We remark that by tracing back through $\A{}$ the \otom{} that yielded this derivation can be found.

%% file: algorithm/P-node.tex
\subsection{P-Node and Q-Node Mapping: Terminology} \label{sec:p-mapping-terms}
Before describing the P-mapping algorithm, we set up some terminology, which is useful both for the P-mapping algorithm and the Q-mapping algorithm (in \cref{sec:q-node-mapping}).

We first define the notion of a partial derivation.
In the Q-mapping and P-mapping algorithms, the derivation of the input node, $x$, is built by considering subsets $U$ of its children. With respect to such a subset $U$, a derivation $\mu$ of $x$ is built as if $x$ had only the children in $U$, and is called a {\em partial derivation}. Formally, $\mu$ is a partial derivation of a node $x$ if $\mu.v=x$ and there is a subset of children $U'\subseteq \children{x}$ such that the two following conditions are true. First, for every $u\in U'$ all the leaves in $T(u)$ are neither mapped nor deleted under $\mu$ - that is, there is no mapping pair $(\ell,y) \in \mu.o$ such that $\ell\in \leaves{u}$. Second, for every $v \in \children{x}\setminus U'$ the leaves in $T(v)$ are either mapped or deleted under $\mu$.
For every $u \in U'$, we say that $u$ is {\em ignored under} $\mu$.
Notice that any derivation is a partial derivation, where the set of ignored nodes ($U'$ above) is empty.
Since all derivations that are computed in a single call to the P-mapping or Q-mapping algorithms have the same start point $i$, it can be omitted (for brevity) from the end point function: thus, we denote $\E(x,k_T,k_S)\doteq L(x,k_T,k_S)$. Then, for a set $U$ of nodes, we define $L(U,k_T,k_S) \doteq \sum_{x\in U}\xspan{x}+k_S-k_T$ and accordingly $\E(U,k_T,k_S)\doteq L(U,k_T,k_S)$.

We now define certain collections of derivations with common properties (such as having the same numbers of deletions and end point).
\begin{definition}\label{def:M}
The collection of all the derivations of every node $u \in U$ to suffixes of $S'[1:\E(U,k_T,k_S)]$ with exactly $k_T$ deletions from the tree and exactly $k_S$ deletions from the string is denoted by $\D{(U,k_T,k_S)}$.
\end{definition}
\begin{definition}\label{def:M-leq}
The collection of all the best derivations from the nodes in $U$ to suffixes of $S'[1:\E(U,k_T,k_S)]$ with up to $k_T$ deletions from the tree and up to $k_S$ deletions from the string is denoted by $\D_\leq{(U,k_T,k_S)}$. 
Specifically, for every node $u \in U$, $k'_T\leq k_T$ and $k'_S\leq k_S$, the set $\D_\leq{(U,k_T,k_S)}$ holds only one highest scoring derivation of $u$ to a suffix of $S'[1:\E(U,k_T,k_S)]$ with $k'_T$ and $k'_S$ deletions from the tree and string, respectively.\footnote{$\D_\leq{(U,k_T,k_S)}$ can be defined using \cref{def:M}: $\D_\leq{(U,k_T,k_S)}=\displaystyle\bigcup_{u \in U}\displaystyle\bigcup_{k'_T\leq k_T}\displaystyle\bigcup_{k'_S\leq k_S}\displaystyle\max_{\substack{\mu \in \D{(U,k_T,k_S)} \\ \mathrm{s.t.}\\ \mu.\dt=k'_T \\ \mu.\ds=k'_S \\ \mu.v=u}}{\mu.score}$.}
\end{definition}

It is important to distinguish between these two definitions. First, the derivations in $\D{(U,k_T,k_S)}$ have {\em exactly} $k_T$ and $k_S$ deletions, while the derivations in $\D_\leq{(U,k_T,k_S)}$ have {\em up to} $k_T$ and $k_S$ deletions. 
Second, in $\D{(U,k_T,k_S)}$ there can be several derivations that differ only in their score and in the \otom{} that yields them, while in $\D_\leq{(U,k_T,k_S)}$, there is only one derivation for every node $u\in U$ and deletion combination pair $(k'_T,k'_S)$.
Note that the end points of all of the derivations are equal.

\cref{def:M} is used for describing the content of an entry of the DP table, where the focus is on the collection of all the derivations of $x$ to $S'$ with exactly $k_T$ and $k_S$ deletions, $\D{(\{x\},k_T,k_S)}$. For simplicity, the abbreviation $\D{(u,k_T,k_S)} = \D{(\{u\},k_T,k_S)}$ is used.
In every step of the P-mapping and Q-mapping algorithms, a different set of derivations of the children of $x$ is examined, thus, \cref{def:M-leq} is used for $U \subseteq \children{x}$. 
In addition, the set of derivations $\D{}$ that is received as input to the algorithms can be described using \cref{def:M-leq} as can be seen in \cref{eq:input-using-M-leq} below. 
In this equation, the union is over all $U\subseteq\children{x}$ because in this way the derivations of all the children of $x$ with {\em every possible end point} are obtained (in contrast to having only $U=\children{x}$, which results in the derivations of all the children of $x$ with the end point $\E(\children{x},k_T,k_S)$).
\begin{equation}\label{eq:input-using-M-leq}
    \D{} = \bigcup_{U \subseteq \children{x}}\bigcup_{k_T\leq d_T}\bigcup_{k_S\leq d_S} \D_\leq{(U,k_T,k_S)}
\end{equation}

In the P-mapping algorithm for $C \subseteq \children{x}$, the notation $\xC{}$ is used to indicate that the node $x$ is considered as if its only children are the nodes in $C$.
Consequentially, the span of $\xC{}$ is defined as $\xspan{\xC{}} \doteq \sum_{c\in C}\xspan{c}$, and the set $\Mxc{}$ (in \cref{def:M} where $U=\{\xC{}\}$) now refers to a set of {\em partial} derivations.

\subsection{P-Node Mapping: The Algorithm}\label{subsec:p-mapping-alg}
Recall that the input consists of an internal P-node $x$, a string $S'$, limits on the number of deletions from the tree $T$ and the string $S'$, $d_T$ and $d_S$, respectively, and a set of derivations $\D{}$ (see \cref{eq:input-using-M-leq}).
The output is $\bigcup_{k_T \leq d_T}\bigcup_{k_S \leq d_S}\argmax_{\mu \in \Mx{}}\mu.score$, which is the collection of the best scoring derivations of $x$ to every possible prefix of $S'$ having up to $d_T$ and $d_S$ deletions from the tree and string, respectively. Thus, there are $O(d_T d_S)$ derivations in the output.
The pseudocode of our algorithm is given in \cref{alg:p-node} in \cref{sec:proofs}.

The algorithm constructs a 3-dimensional DP table $\PP{}$, which has an entry for every $0\leq k_T \leq d_T$, $0\leq k_S \leq d_S$ and subset $C\subseteq \children{x}$. The purpose of an entry $\PP[C,k_T,k_S]$ is to hold the best score of a partial derivation in $\Mxc{}$, i.e. a partial derivation rooted in $\xC{}$ to a prefix of $S'$ with exactly $k_T$ deletions from the tree and $k_S$ deletions from the string. The children of $x$ that are not in $C$ are {\em ignored} (as defined in \cref{sec:p-mapping-terms}) under the partial derivation stored by the DP table entry $\PP[C,k_T,k_S]$, thus they are neither deleted nor counted in the number of deletions from the tree, $k_T$. (They will be accounted for in the computation of other entries of $\PP{}$.)
Similarly to the main algorithm, some of the entries of $\PP{}$ are invalid, and their value is defined as $-\infty$ (for more information see \cref{sec:q-node-mapping}).
For lack of space, the description of the initialization of $\PP{}$ is deferred to \cref{sec:p-mapping-init}.

After the initialization, the remaining entries of $\PP{}$ are calculated using the recursion rule in \cref{eq:p-recursion} below. The order of computation is ascending with respect to the size of the subsets $C$ of the children of $x$, and for a given $C\subseteq \children{x}$, the order is ascending with respect to the number of deletions from both tree and string.
\begin{equation} \label{eq:p-recursion}
    \PP[C,k_T,k_S] = 
            \max \begin{cases}
                \PP[C,k_T,k_S-1] \\
                \displaystyle\max_{\mu \in \Mc} \PP[C\setminus \{\mu.v\},k_T-\mu.\dt,k_S-\mu.\ds] + \mu.score\\
            \end{cases}
\end{equation}
Intuitively, every entry $\PP[C,k_T,k_S]$ defines some index $i$ of $S'$ that is the end point of every partial derivation in $\Mxc{}$. 
Thus, $S'[i]$ must be a part of any partial derivation $\mu\in\Mxc{}$, so, either $S'[i]$ is deleted under $\mu$ or it is mapped under $\mu$. The former option is captured by the first case of the recursion rule.
If $S'[i]$ is mapped under $\mu$, then due to the hierarchical structure of $T(x)$, it must be mapped under some derivation $\mu'$ of one of the children of $x$ that are in $C$. 
Thus we receive the second case of the recursion rule.
We remark that the case of a node deletion is captured by the initialization (further explanation can be found in \cref{par:p-node-deletion-intuition}).

Once the entire DP table is filled, a derivation of maximum score for every end point and deletion number combination can be found in $\PP[\children{x}, \cdot, \cdot]$. For the output derivations to be ordered with respect to their end point, they need to be extracted by traversing $\PP[\children{x}, \cdot, \cdot]$ in the order described in \cref{sec:q-node-mapping} and exemplified in \cref{tab:deletion-traversal-order}.

The time complexity analysis of the algorithm can be found in \cref{proof:p-time}, and the proof of correctness can be found in \cref{sec:p-node-proof}.

%% file: algorithm/time-complex.tex
\subsection{Complexity Analysis of the Main Algorithm}
In this section we compare the time complexity of the main algorithm (in \cref{sec:general-algorithm}) to the na\"ive solution for \up{}. We note that the proof of correctness of the algorithm can be found in \cref{sec:main-alg-proof}. The proof of \cref{lemma:up-time} below is given in \cref{proof:general-time}.

\begin{lemma}\label{lemma:up-time}
The algorithm in \cref{sec:general-algorithm} runs in $O(n \gamma {d_T}^2 {d_S}^2 (m_p 2^{\gamma} + m_q))$ time and $O(d_T d_S (m n + 2^\gamma))$ space, where $\gamma$ is the maximum degree of a node in $T$.
\end{lemma}

Thus, it is proven that \up{} has an \FPT{} solution with the parameter $\gamma$ (\cref{theorem:up-fpt}).
\begin{theorem}\label{theorem:up-fpt}
\up{} with parameter $\gamma$ is \FPT{}. Particularly, it has an \FPT{} algorithm that runs in $O^*(2^{\gamma})$ time\footnote{The notation O* is used to hide factors polynomial in the input size.}.
\end{theorem}

The na\"ive solution for \up{} and its time complexity analysis are given in \cref{sec:naive}. There we show that it solves \up{} in $O(2^{m_q}{(\gamma!)}^{m_p} n m(d_T+d_S)d_T d_S)$ time.
We conclude that the time complexity of our algorithm is substantially better, exemplified by considering two complementary cases. 
One, when there are only P-nodes in $T$ (i.e. $m=m_p$), the na\"ive algorithm is super-exponential in $\gamma$, and even worse, exponential in $m$, while ours is exponential only in $\gamma$, and hence polynomial for any $\gamma$ that is constant (or even logarithmic in the input size). 
Second, when there are only Q-nodes in $T$ (i.e. $m=m_q$), the na\"ive algorithm is exponential while ours is polynomial.

%% file: methods.tex
\section{Methods and Datasets}
\label{sec:methods}
\para{Dataset and Gene Cluster Generation.}
$1,487$ fully sequenced prokaryotic strains with COG ID annotations were downloaded from GenBank (NCBI; ver 10/2012). Among these strains, 471 genomes included a total of 933 plasmids. 

The gene clusters were generated using the tool CSBFinder-S \cite{csbfinder-s}. 
CSBFinder-S was applied to all the genomes in the dataset after removing their plasmids, using parameters $q=1$ (a colinear gene cluster is required to appear in at least one genome) and $k=0$ (no insertions are allowed in a colinear gene cluster), resulting in 595,708 colinear gene clusters. 
Next, ignoring strand and gene order information, colinear gene clusters that contain the exact same COGs were united to form the generalized set of gene clusters. The resulting gene clusters were then filtered to 26,270 gene clusters that appear in more than 30 genomes. 

\smallskip
\para{Generation of \PQT{s}.}
The generation of \pqt{s} was performed using a program \cite{lev-tool} that implements the algorithm described in \cite{landau2005gene} for the construction of a \pqt{} from a list of strings comprised from the same set of characters. In the case where a character appeared more than once in a training string, the \pqt{} with the minimum consistent frontier size was chosen.
The generated \pqt{s} varied in size and complexity. The length of their frontier ranged between $4$ and $31$, and the size of their consistent frontier ranged between $4$ and $362,880$. 
 
\smallskip
\para{Implementation and Performance.} \alg{} is implemented in Java 1.8. The runs were performed on an Intel Xeon X5680 machine with 192 GB RAM. The time it took to run all plasmid genomes against one \pqt{} ranged between $5.85$ seconds (for a \pqt{} with a consistent frontier of size $4$) and $181.5$ seconds (for a \pqt{} with a consistent frontier of size $362,880$).
In total it took an hour and 47 minutes to run every one of the $779$ \pqt{s} against every one of the $933$ plasmids.

\smallskip
\para{Substitution Scoring Function.}
The substitution scoring function reflects the distance between each pair of COGs, that is computed based on sentences describing the functional annotation of the COGs (e.g., "ABC-type sugar transport system, ATPase component"). The "Bag of Words model" was employed, where the functional description of each COG is represented by a sparse vector that is normalized to have a unit Euclidean norm. First, each COG description was tokenized and the occurrences of tokens in each description was counted and normalized using tf–idf term weighting. Then, the cosine similarity between each two vectors was computed, resulting in similarity scores ranging between 0 and 1. The sentences describing COGs are short, therefore each word largely influences the score, even after the tf–idf term weighting. Therefore, words that do not describe protein functions that were found in the top 30 most common words in the description of all COGs were used as stop-words. 
Two COGs with the same COG IDs were set to have a score of 1.1, and the substitution score between a gene with no COG annotation to any other COG was set to be -0.1. Two COGs with a zero score were penalized to have a score of -0.2 and the deletion of a COG from the query or the target string was set to have a score of zero.

\smallskip
\para{Enrichment Analysis.} For each of the four variants in \cref{fig:rnd_pqt}.C, a hypergeometric test was performed to measure the enrichment of the corresponding variant in one of the classes in which it appears. A total of 10 p-values were computed and adjusted using the Bonferroni correction; two p-values were found significant (<$0.05$), reported in \cref{sec:results}.

\smallskip
\para{Specificity Score.}
We define a specificity score for a \pqt{} $T$ of a gene cluster named S-score. Let $\tilde{T}$ be the least specific \pqt{} that could have been generated for the genes of the gene cluster based on which $T$ was constructed. Namely, a \pqt{} that allows all permutations of said genes, has height $1$, is rooted in a P-node whose children (being the leaves of the tree) are the leaves of $T$. Thus, the S-score of $T$ is $\frac{|C(\tilde{T})|}{|C(T)|}$.
For a gene cluster of permutations (i.e. there are no duplications), the computation of $|C(T)|$ is as described in \cref{eq:CT-size}, where the set of P-nodes in $T$ is denoted by $T.p$. 
\begin{equation}\label{eq:CT-size}
    |C(T)| = 2^{m_q}\cdot \prod_{x\in T.p}{|\children{x}|!}
\end{equation}
For a gene cluster that has duplications, the set $C(T)$ is generated to learn its size.
Let $\mathsf{a}(\ell,T)$ denote the number of appearances of the label $\ell$ in the leaves of $T$ and let $\mathsf{labels}(T)$ denote the set of all labels of the leaves of $T$. So, the formula for $|C(\tilde{T})|$ is as in \cref{eq:CPnode}. Clearly, for $T$ with no duplications $|C(\tilde{T})| = |F(T)|!$.
\begin{equation}\label{eq:CPnode}
    |C(\tilde{T})| = \frac{|F(T)|!}{\prod_{\ell\in \mathsf{labels}(T)}{\mathsf{a}(\ell,T)!}}
\end{equation}

%% file: Bio-results.tex
\section{Results}
\label{sec:results}
\subsection{Chromosomal Gene Orders Rearranged in Plasmids}
\label{sec:results_shuffling}
The labeling of each internal node of a \pqt{} as P or Q, is learned during the construction of the tree, based on some interrogation of the gene orders from which the \pqt{} is trained \cite{landau2005gene}. As a result, the set of strings that can be derived from a \pqt{} $T$, consists of two parts: (1) all the strings representing the known gene orders from which $T$ was constructed, and (2) additional strings, denoted {\em tree-guided rearrangements}, that do not appear in the set of gene orders constructing $T$, but can be obtained via rearrangement operations that are constrained by $T$. Thus, the tree-guided rearrangements conserve the internal topology properties of the gene cluster, as learned from the corresponding gene orders during the construction of $T$, such that colinear dependencies among genes and between sub-operons are preserved in the inferred gene orders.

In this section, we used the \pqt{s} constructed from chromosomal gene clusters, to examine whether tree-guided rearrangements can be found in plasmids.
The objective was to discover gene orders in plasmids that abide abide by a \pqt{} representing a chromosomal gene cluster, and differ from all the gene orders participating in the \pqt{'s} construction. 
\pqt{s} that are constructed from gene clusters that have only one gene order or gene clusters with less than four COGs cannot generate gene orders that differ from the ones participating in their construction. Therefore, only 779 out of 26,270 chromosomal gene clusters were used for the construction of query \pqt{s} (the generation of the chromosomal gene clusters is detailed in \cref{sec:methods}). 
Using our tool \alg{} that implements the algorithm proposed for solving the \up{} problem, the query \pqt{s} were run as queries against all plasmid genomes. This benchmark was run conservatively without allowing substitutions or deletions from the \pqt{} or from the target string. 380 of the query gene clusters were found in at least one plasmid. The instances of these gene clusters in plasmids are provided in the Supplementary Materials as a session file that can be viewed using the tool CSBFinder-S \cite{csbfinder-s}.
 
Tree-guided rearrangements were found among instances of 29 gene clusters.
The \pqt{s} corresponding to these gene clusters were sorted by a decreasing S-score, where higher scores are given to a more specific tree (details in Section \ref{sec:methods}). In this setting, the higher the S-score, the smaller the number of possible gene orders that can be derived from the respective \pqt{}. Interestingly, 21 out of these 29 gene clusters code for transporters, namely 20 importers (ABC-type transport systems) and one exporter (efflux pump). The 10 top ranking results are presented in \cref{table:shuffling}.

We selected the third top-ranking \pqt{} in \cref{table:shuffling} for further analysis. This \pqt{} was constructed from 7 gene orders of a gene cluster that encodes a heavy metal efflux pump. This gene cluster was found in the chromosomes of 79 genomes (represented by the 7 distinct gene orders mentioned above) and in the plasmids of 7 genomes. The tree-guided rearrangement instance was found in the strain \textit{Cupriavidus metallidurans CH34}, isolated from an environment polluted with high concentrations of several heavy metals. This strain contains two large plasmids that confer resistance to a large number of heavy metals such as zinc, cadmium, copper, cobalt, lead, mercury, nickel and chromium. We hypothesize that the rearrangement event could have been caused by a heavy metal stress \cite{doi:10.1080/1040841X.2017.1303661}. In the following section we will focus on this \pqt{} to further study its different variants in plasmids.

\begin{table}[ht!]
\centering
\begin{tabular}{lllll}
\hline
{} &                                  PQ-tree$^1$ &   S-score & \# Genomes$^2$ &                                                                              Functional Category \\
\hline
1  &  [[0683 [[0411 0410] [0559 4177]]] 0583] &      22.5 &  5 (2) &  Amino acid transport \\
2  &  (1609 [1653 1175 0395] 3839) &      10.0 &  10 (2) &  Carbohydrate transport \\
3  &  [[1538 [3696 0845]] [0642 0745]] &       7.5 &  7 (1) &  Heavy metal efflux\\
4  &  [[2115 1070] [4213 [1129 4214]]] &       7.5 &  1 (1) &  Carbohydrate transport \\
5  &  [1960 [[2011 1135] [2141 1464]]] &       7.5 &  3 (1) &  Amino acid transport \\
6  &  [[0596 0599] [[3485 3485] 0015]] &       7.5 &  9 (1) &  Metabolism \\
7  &  [[[1129 1172 1172] 1879] 3254] &       7.5 &  6 (1) &  Carbohydrate transport \\
8  &  (1609 1869 [[1129 1172] 1879] 0524) &       7.5 &  1 (1) &  Carbohydrate transport \\
9  &  (0683 [0559 4177] [0411 0410] 0318) &       7.5 &  1 (1) &  Amino acid transport \\
10 &  (3839 0673 [[0395 1175] 1653]) &       5.0 &  10 (1) &  Carbohydrate transport \\
\hline

\end{tabular}
\caption{Ten top ranked \pqt{s} for which tree-guided rearrangements were found in plasmids. $^1$Square brackets represent a Q-node; round brackets represent a P-node. Numbers indicate the respective COG IDs. $^2$This column indicates the number of genomes harboring plasmid instances of the respective \pqt. The number in brackets indicates the number of genomes harboring a tree-guided gene rearrangement of the corresponding gene cluster. The full table can be found in~\cref{table:shuffling_full}.}
\label{table:shuffling}
\end{table}

\subsection{RND Efflux Pumps in Plasmids}
\label{sec:results_rnd}
The heavy metal efflux pump examined in the previous section (corresponding to the third top-ranking \pqt{} in \cref{table:shuffling}), was used as a \alg{} query and re-ran against all the plasmids in our dataset in order to discover approximate instances of this gene cluster, possibly encoding remotely related variations of the efflux pump it encodes. This time, in order to increase sensitivity, a semantic substitution scoring function (described in \cref{sec:methods}) was used, and the parameters were set to $d_T=1$ (up to one deletion from the tree, representing missing genes) and $d_S=3$ (up to three deletions from the plasmid, representing intruding genes).
An instance of a gene cluster is accepted if it was derived from the corresponding \pqt{} with a score that is higher than 0.75 of the highest possible score attainable by the query.
The plasmid instances detected by \alg{} are displayed in \cref{fig:pump_instances}. 

Heavy metal efflux pumps are involved in the resistance of bacteria to a wide range of toxic metal ions \cite{nies2003efflux} and they belong to the resistance-nodulation-cell division (RND) family.
In Gram-negative bacteria, RND pumps exist in a tripartite form, comprised from an outer-membrane protein (OMP), an inner membrane protein (IMP), and a periplasmic membrane fusion protein (MFP) that connects the other two proteins. In some cases, the genes of the RND pump are flanked with two regulatory genes that encode the factors of a two-component regulatory system comprising a sensor/histidine kinase (HK) and response regulator (RR) (\cref{fig:rnd_pqt}.B). This regulatory system responds to the presence of a substrate, and consequently enhances the expression of the efflux pump genes.

The \pqt{} of this gene cluster (\cref{fig:rnd_pqt}.A) shows that the COGs encoding the IMP and MFP proteins always appear as an adjacent pair, the OMP COG is always adjacent to this IMP-MFP pair, and the HK and RR COGs appear as a pair downstream or upstream to the other COGs. COG3696, which encodes the IMP protein, is annotated as a heavy metal efflux pump protein, while the other COGs are common to all RND efflux pumps. Therefore, it is very likely that the respective gene cluster corresponds to a heavy metal RND pump. The absence of an additional periplasmic protein likely indicates that this gene cluster encodes a Czc-like efflux pump that exports divalent metals such as the cobalt, zinc and cadmium exporter in \textit{Cupriavidus metallidurans} \cite{nies2003efflux} (\cref{fig:rnd_pqt}.C(1)). 

\alg{} discovered instances of this gene cluster in the plasmids of 12 genomes (Figures \ref{fig:rnd_pqt}.C(1) and \ref{fig:rnd_pqt}.D), and it is significantly enriched in the $\beta$-proteobacteria class (hypergeometric p-value= $1.09\times10^{-5}$, Bonferroni corrected p-value = $1.09\times10^{-4}$). In addition, three other variants of RND pumps were found as instances of the query gene cluster (\cref{fig:rnd_pqt}.C(2-4)). The plasmids of three genomes contained instances that were missing the COG corresponding to the OMP gene CzcC (\cref{fig:rnd_pqt}.C(2)). This could be caused by a low quality sequencing or assembly of these plasmids. An alternative possible explanation is that a Czc-like efflux pump can still be functional without CzcC; a previous study showed that the deletion of CzcC resulted in the loss of cadmium and cobalt resistance, but most of the zinc resistance was retained \cite{nies2003efflux}. 

Some instances identified by the query, found in the plasmids of six genomes, seem to encode a different heavy metal efflux pump (\cref{fig:rnd_pqt}.C(3)). This variant includes all COGs from the query, in addition to an intruding COG that encodes a periplasmic protein (CusF). This protein is a predicted copper usher that facilitates access of periplasmic copper towards the heavy metal efflux pump. Indeed, the genomic region of Cus-like efflux pumps that export monovalent metals, such as the silver and copper exporter in \textit{Escherichia coli}, include this periplasmic protein, in contrast to the Czc-like efflux pump \cite{nies2003efflux}. This variant was found in the plasmids of six bacterial genomes belonging to the class $\gamma$-proteobacteria (\cref{fig:rnd_pqt}.D). This gene cluster is significantly enriched in the $\gamma$-proteobacteria class (hypergeometric p-value= $2.13\times10^{-4}$, Bonferroni corrected p-value = $2.13\times10^{-3}$). 
Surprisingly, all of these strains, except for one, 
are annotated as human or animal pathogens. Interestingly, previous studies suggest that the host immune system exploits excess copper to poison invading pathogens \cite{fu2014copper}, which can explain why these pathogens evolved copper efflux pumps. 

Another variant of the pump, appearing in five genomes (Figures \ref{fig:rnd_pqt}.C(4) and \ref{fig:rnd_pqt}.D), resulted from a substitution of the query IMP gene (COG3696) by a different IMP gene (COG0841) belonging to the multidrug efflux pump AcrAB/TolC. 
The AcrAB-TolC system, mainly studied in \textit{Escherichia coli}, transports a diverse array of compounds with little chemical similarity \cite{du2014structure}. AcrAB/TolC is an example of an intrinsic non-specific efflux pump, which is widespread in the chromosomes of Gram-negative bacteria, and likely evolved as a general response to environmental toxins \cite{sulavik2001antibiotic}. In this case, the query gene cluster and the identified variant share all COGs, except for the COGs encoding the IMP genes. The differing COGs are responsible for substrate recognition, which naturally differs between the two pumps, as one pump exports heavy metal while the other exports multiple drugs.
When considering the functional annotation of these two COGs, we see that the query metal efflux pump COG encoding the IMP gene is annotated as "Cu/Ag efflux pump CusA", while in the multidrug efflux pump the COG encoding the IMP gene is annotated as "Multidrug efflux pump subunit AcrB".
Thus, in spite of the difference in substrate specificity, the semantic similarity measure employed by \alg{} was able to reflect their functional similarity and allowed the substitution between them, while conferring to the structure of the \pqt{}. 

\begin{figure}[t]
	\centering
	\includegraphics[width=\linewidth]{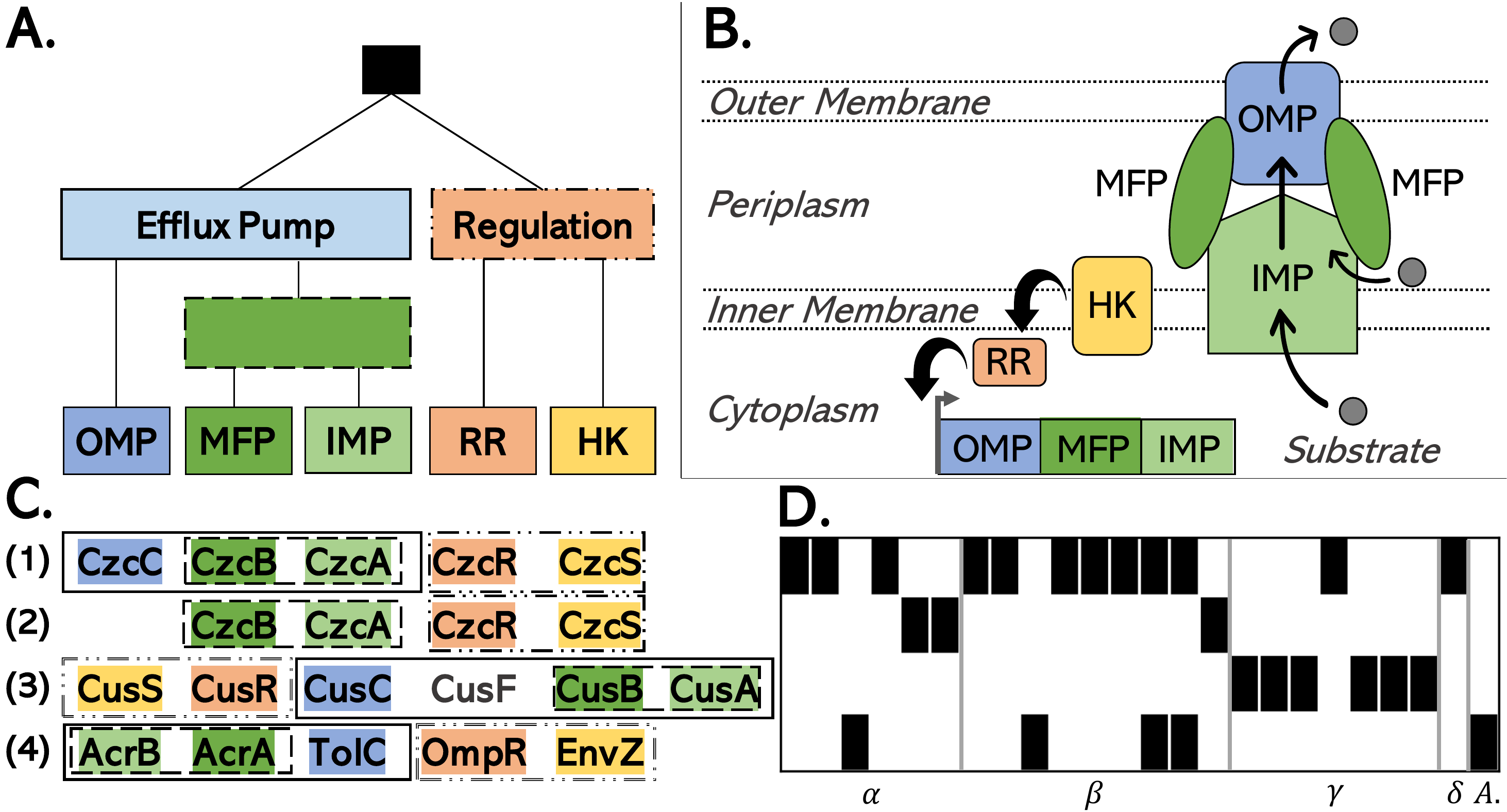}
	\caption{\textbf{A.} A \pqt{} of a heavy metal RND efflux pump, corresponding to the third top scoring result in \cref{table:shuffling}. \textbf{B.} An illustration of an RND efflux pump consisting of an outer-membrane protein (OMP), an inner membrane protein (IMP), and a periplasmic membrane fusion protein (MFP) that connects the other two proteins. In addition, a two-component regulatory system consisting of a sensor/histidine kinase (HK) and response regulator (RR) enhances the transcription of the efflux pump genes. \textbf{C.} Representatives of the three different RND efflux pumps found in plasmids. \textbf{(1)} A Czc-like heavy metal efflux pump, \textbf{(2)} A Czc-like heavy metal efflux pump with a missing OMP gene, \textbf{(3)} A Cus-like heavy metal efflux pump, \textbf{(4)} An Acr-like multidrug efflux pump. Additional details can be found in the text. \textbf{D.} The presence-absence map of the three types of efflux pumps found in the plasmids of different genomes. The rows correspond to the rows in (C), the columns correspond to the genomes in which instances were found, organized according to their taxonomic classes. A black cell indicates that the corresponding efflux pump is present in the plasmids of the genome. The labels below the map indicate the classes $\alpha,\beta,\gamma,\delta$-Proteobacteria and Acidobacteriia.}
\label{fig:rnd_pqt}
\end{figure}

%% file: conclusion.tex
\section{Conclusions}
In this paper, we defined a new problem in comparative genomics, denoted \up{}. The objective of \up{} is to identify approximate new instances of a gene cluster in a new genome $S$. In our model, the gene cluster is represented by a \pqt{} $T$, and the approximate instances can vary from the known gene orders by genome rearrangements that are constrained by $T$, by gene substitutions that are governed by a gene-to-gene substitution scoring function $h$, and by gene deletions and insertions that are bounded from above by integer parameters $d_T$ and $d_S$, respectively.

We proved that the \up{} problem is \NPH{} and proposed a parameterized algorithm that solves it in $O^*(2^{\gamma})$ time, where $\gamma$ is the maximum degree of a node in $T$ and $O^*$ is used to hide factors polynomial in the input size.

The proposed algorithm was implemented as a publicly available tool and harnessed to search for tree-guided rearrangements of chromosomal gene clusters in plasmids. We identified 29 chromosomal gene clusters that are rearranged in plasmids, where the rearrangements are guided by the corresponding PQ-tree.
One of those gene clusters, coding for a heavy metal efflux pump, was further analysed to characterize its approximate instances in plasmids. An interesting variant of the analysed gene cluster, found among its approximate instances, corresponds to a copper efflux pump. It was found mainly in pathogenic bacteria, and likely constitutes a bacterial defense mechanism against the host immune response.
These results exemplify how our tool can be harnessed to find meaningful variations of known biological systems that are conserved as gene clusters, suggesting that \up{} can be further utilized in the domain of comparative functional analysis.

One of the downsides to using \pqt{}s to represent gene clusters is that very rare gene orders taken into account in the tree construction could greatly increase the number of allowed rearrangements and thus substantially lower the specificity of the \pqt{}. 
Thus, a natural continuation of our research would be to increase the specificity of the model by considering a stochastic variation of \up{}. Namely, defining a \pqt{} in which the internal nodes hold the probability of each rearrangement, and adjusting the algorithm for \up{} accordingly. In addition, future extensions of this work could also aim to increase the sensitivity of the model by taking into account gene duplications, gene-merge and gene-split events, which are typical events in gene cluster evolution.

%% file: NPH.tex
\section{\PQT{} Search is NP-Hard} \label{sec:NPh}
In this section we prove \cref{theorem:up-nph} by describing a reduction from the \jisp{} (\JISP{}) to \up{.}
\begin{theorem}\label{theorem:up-nph}
\up{} is \NPH{.}
\end{theorem}
\JISP{} was introduced by Nakajima and Hakimi \cite{nakajima1982complexity}. They considered one machine and a collection of non-preamble jobs, denoted $1,\dotso{},n$, that need to be executed on that machine. Each job $i$ has an execution time $t_i$ and $k_i$ possible starting times, $(s_{i_1},\dotso,s_{i_{k_i}})$. Note that every $t_i$ and $s_{i_j}$ define an interval on the real line: $[s_{i_j},s_{i_j}+t_i]$. The aim is to allocate a starting time for each job such that no two jobs will run simultaneously on the machine. The \jisp{} (\JISP{}) with $k$ intervals per job was named \JISP{$k$} \cite{spieksma1999approximability}.

Since its initial definition, the problem has seen many equivalent definitions \cite{keil1992complexity,spieksma1999approximability,spieksma1992complexity,vanBevern2015}. We use the following formulation for \JISP{$k$} based on colors. In this setting, each job $i$ is encoded as a $k$-tuple of intervals on the real line having the color $i$. Let $\gamma$ be the number of colors, hence there are $\gamma$ jobs to be executed. The notation $I_j^i$ is used to denote the interval with starting time $s_{ij}$ finishing time $f_{ij}$ (i.e. duration $[s_{ij},f_{ij}]$) and color $1\leq i\leq \gamma$ (i.e. it is a part of the $i^{\textrm{th}}$ $k$-tuple). The objective is to select exactly one interval of each color ($k$-tuple) such that no two intervals intersect.

\JISP{3} was shown to be \NPC{} by Keil \cite{keil1992complexity}. Crama et al. \cite{spieksma1992complexity} showed that \JISP{3} is \NPC{} even if all intervals are of length 2. We use these results to show that \up{} is \NPH{.}

\subparagraph{The Reduction.} Given an instance, $J$, of JISP3 where all intervals have length 2, an instance of \up{} is created. It is easy to see that shifting all intervals by some constant does not change the problem. Hence, assume that the leftmost starting interval starts at $1$. Let $L$ be the rightmost ending point of an interval, so the focus can be only on the segment $[1,L]$ of the real line. 
Now, an instance of \up{} $(T,S,h,d_T,d_S)$ is constructed (an illustrated example is given in \cref{fig:reduction} below):
\begin{itemize}
    \item \textbf{The \pqt{} $T$:} The root node, $root$, is a P-node with $3L-2-3\gamma$ children: $x_1, \dotso{,} x_{\gamma},\allowbreak y_1, \dotso{,} y_{3L-2-4\gamma}$. The children of $root$ are defined as follows: for every color $1 \leq i \leq \gamma$, create a Q-node $x_i$ with four children $x_i^s, \ x_i^a, \ x_i^b, \ x_i^f$; for every index $1\leq i\leq 3L-2-3\gamma$, create a leaf $y_i$.
    \item \textbf{The string $S$:} Define $S=\sigma_1 \sigma_a \sigma_b \sigma_2 \sigma_a \sigma_b \dotso \sigma_a \sigma_b \sigma_L$.
    \item \textbf{The substitution function $h$:} for every interval of the color $i$, $I_j^i=[s_{ij},f_{ij}]$, the function $h$ returns $True$ for the following pairs: $(x_i^s, \sigma_{s_{ij}})$, $(x_i^f, \sigma_{f_{ij}})$, $(x_i^a, \sigma_a)$ and $(x_i^b, \sigma_b)$. In addition, every leaf $y_r$ can be substituted by every letter of $S$, namely for every index $1\leq r\leq 3L-2-3\gamma$ and for every $s \in \{a,b,1,\dotso,L\}$ the function $h$ returns $True$ for the pair $(y_r, \sigma_s)$. For every other pair $h$ returns $False$. For the optimization version of the problem, define a scored substitution function $h'$, such that $h'(u,v) = 1$ if $h(u,v) = True$ and $h'(u,v) = -\infty$ if $h(u,v) = False$.
    \item \textbf{Number of deletions:} Define $d_T=0$ and $d_S=0$, i.e. deletions are forbidden from both tree and string.
\end{itemize}

An example of the reduction is shown in \cref{fig:reduction}. A collection of two $3$-tuples (one blue and one red) where each interval is of length 2, i.e a \JISP{3} instance, is in \cref{fig:JISP-instance}. Running the reduction algorithm yields the \up{} instance in \cref{fig:up-instance}. The pairs that can be substituted (i.e. the pairs for which $h$ returns $True$) are given by the lines connecting the leafs of the \pqt{} and the letters of the string $S$. The nodes and substitutable pairs created due to the blue and red intervals in the \JISP{3} instance are marked in blue and red, respectively. The substitutable pairs containing a $y$ node are marked in gray. Note that the colors given in \cref{fig:up-instance} are not a part of the \up{} instance, and are given for convenience. 

\begin{figure}
    \centering
    \begin{subfigure}[b]{0.7\textwidth}
        \centering
        \includestandalone[width=0.8\textwidth]{figures/reduction_interval}
        \caption{}
        \label{fig:JISP-instance}
    \end{subfigure}
    \par\bigskip
    \begin{subfigure}[b]{1\textwidth}
        \centering
        \includegraphics[width=0.8\textwidth]{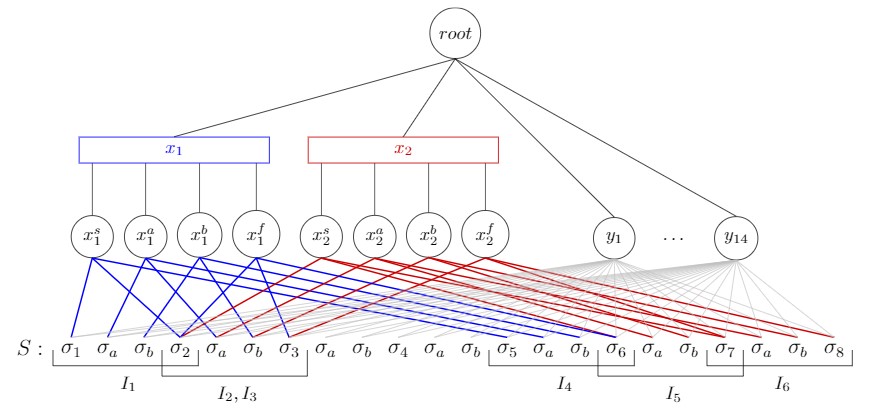}
        \caption{}
        \label{fig:up-instance}
    \end{subfigure}
    \caption{\textbf{(a)} The input of the reduction - a \JISP{3} instance $J$ with intervals of length 2. \textbf{(b)} The output of the reduction - a \up{} instance $(T,S,h,d_T,d_S)$.}
    \label{fig:reduction}
\end{figure}

\begin{proof}[Correctness]
Let $J$ be an instance of \JISP{3}, and let $(T,S,h,d_T,d_S)$ be the output of the reduction on this instance. We prove that there exists a collection of intervals that is a solution for $J$ if and only if there exists a one-to-one mapping that is a solution to $(T,S,h,d_T,d_S)$.
\subparagraph{One Direction.}
Suppose that there exists a solution to the output instance of \up{} of the reduction, $(T,S,h,d_T,d_S)$. This solution is a one-to-one mapping $\M{}$: for every $1\leq i\leq\gamma$, a set of pairs of the form $(x_i^j,\sigma_k(\ell))$ for $j\in\{s,f,a,b\}$, and for every $1\leq r\leq 3L-2-3\gamma$, pairs of the form $(y_r,\sigma_k(\ell))$ where $k\in \{1,\dotso,L,a,b\}$ and $1\leq \ell\leq 3L-2$. By the definition of \up{}, each $x_i^j$, $y_r$ and $\sigma_k(\ell)$ appear in exactly one pair.
Considering the mappings of the children of a node $x_i$, they must be the following: $(x_i^s,\sigma_k(\ell))$, $(x_i^a,\sigma_a(\ell+1))$, $(x_i^b,\sigma_b(\ell+2))$ and $(x_i^f,\sigma_{k+1}(\ell+3))$. 
To see this, observe that a node $x_i^a$ must be mapped to $\sigma_a$, because it is the only letter by which it can be substituted under $h$. 
In the same way, a node $x_i^b$ must be mapped to $\sigma_b$. Because $d_T=0$, $d_S=0$ and due to the properties of a Q-node, once $x_i^s$ is mapped to the letter in index $\ell$ (i.e. $(x_i^s,\sigma(\ell))\in \M{}$), $x_i^a$ must be mapped to the letter in index $\ell+1$ or in index $\ell-1$ (i.e. the adjacent letter to the one to which $x_i^s$ is mapped), then $x_i^b$ must be mapped to the letter in index $\ell+2$ or $\ell-2$, respectively, and $x_i^f$ to $\ell + 3$ or $\ell -3$, respectively. Since $\sigma_a$ is always the letter preceding $\sigma_b$ in $S$, $x_i^b$ must be mapped to an index larger by one than the index mapped to $x_i^a$. Hence, the children of the Q-node $x_i$ are mapped from left to right.

Now, let us derive a solution for the original \JISP{3} instance from the solution to \up{}. For every $3$-tuple of color $1\leq i\leq\gamma$, where $(x_i^s,\sigma_k(\ell))\in$ $\M{}$, choose the interval $I_k^i=[k,k+1]$ from the $3$-tuple of color $i$. For example, if a part of the solution for the \up{} instance in \cref{fig:up-instance} is $\{(x_1^s,\sigma_1(1)),$ $(x_1^a,\sigma_a(2)),$ $(x_1^b,\sigma_b(3)),$ $(x_1^f,\sigma_2(4))\} \subset$ $\M{}$, then $I_1^1$ is the interval chosen for the first color (blue) in the derived solution for the \JISP{3} instance in \cref{fig:JISP-instance}. Note that $I_k$ is indeed one of the intervals of color $i$, due to the definition of $h$,
$h(x_i^s,\sigma_k) = True$ and $h(x_i^f,\sigma_{k+1})=True$ if and only if there is an interval of color $i$ starting at $k$ and ending at $k+1$. Thanks to $\M{}$ being a one-to-one mapping, the intervals do not intersect, and for every color there is only one interval chosen.

\subparagraph{Second Direction.}
Let us prove that if there is a solution for the original instance of \JISP{3} $J$, then there is a solution for $(T,S,h,d_T,d_S)$. Let $\mathcal{I}=\{I_{j_1}^1,...,I_{j_\gamma}^\gamma\}$ be a solution of $J$ such that $I_{j_i}^i=[s_{ij_i},f_{ij_i}]$ is the interval chosen for the $3$-tuple of color $i$. First, the solution for the \up{} instance $(T,S,h,d_T,d_S)$ is constructed. For every $1\leq i\leq\gamma$, insert the following pairs into $\M{}$: $(x_i^s,\sigma_{s_{ij_i}}(3s_{ij_i}-2))$, $(x_i^a,\sigma_{a}(3s_{ij_i}-1))$, $(x_i^b,\sigma_{b}(3s_{ij_i}))$, and $(x_i^f,\sigma_{f_{ij_i}}(3f_{ij_i}-2))$.  
For example, if $I_2^2$ is the interval chosen from the second (red) $3$-tuple in the solution of the \JISP{3} instance in \cref{fig:JISP-instance}, then the solution for the \up{} instance in \cref{fig:up-instance} includes the pairs $\{(x_2^s,\sigma_2(4)), (x_2^a,\sigma_a(5)), (x_2^b,\sigma_b(6)), (x_2^f,\sigma_3(7))\}$. 
Observe that only one pair was inserted for every leaf of $T$, and since no two intervals intersect, every index of $S$ appears in only one pair in $\M{}$. Hence, a one-to-one mapping between $4\gamma$ leafs of $T$ and $4\gamma$ indices of $S$ was defined, and $3L-4\gamma-2$ additional pairs need to be inserted to $\M{}$ in order to construct a solution for the \up{} instance. According to $h$, every node $y_r$ ($1\leq r\leq 3L-2-3\gamma$) can be mapped to every letter $\sigma_k$, so arbitrarily insert the pairs $(y_r,\sigma_{k_r}(\ell_r))$ to $\M{}$, such that no index or node appear in more than one pair. (It can be done because there are $3L-4\gamma-2$ $y$ nodes and after mapping the $4$ children of every one of the $\gamma$ $x_i$ nodes, $3L-4\gamma-2$ characters of $S$ are left without a mapping). Thus, a one-to-one mapping $\M{}$ between all the leafs of $T$ and all the indices of $S$ (i.e. no deletions from $S$ and $T$) was defined, and it is left to prove that $S$ can be derived from $T$ under $\M{}$. 

The children of a Q-node $x_i$ from left to right are: $x_i^s,x_i^a,x_i^b,x_i^f$, and so, because $d_T=0$ and $d_S=0$ (no deletions from both tree and string), they have to be mapped to consecutive indices of $S$; this is indeed the case according to our definition of $\M{}$. The mapping of every $y_r$ is obviously also legal. Finally, $root$ is a P-node, so its children can be arranged in any order, and they are. This completes the proof of correctness of the reduction.
\end{proof} 

This concludes the proof of \cref{theorem:up-nph}.

\subparagraph{The Importance of $\sigma_a$ and $\sigma_b$ in the Reduction.}\label{sec:append}
In the reduction from \JISP{3} to \up{} the string $S$ was defined such that there is a character $\sigma_i$ for every $1\leq i \leq L$, and between every two such characters there is the sequence $\sigma_a\sigma_b$, i.e. $S=\sigma_1 \sigma_a \sigma_b \sigma_2 \sigma_a \sigma_b \dotso \sigma_a \sigma_b \sigma_L$. 
In addition the \pqt{} $T$ and the substitution function $h$ were defined such that for every color $i$ ($1\leq i\leq\gamma$), there are the leafs $x_i^a$ and $x_i^b$ in $T$ and $h$ returns $True$ for both $(x_i^a,\sigma_a)$ and $(x_i^b,\sigma_b)$. 
For abbreviation these leafs, the multiple appearances of $\sigma_a\sigma_b$ in $S$ and the allowed substitutions between them are named \abs{}. Here we explain why the \abs{} is important.

The necessity arises when considering the first direction of the reduction, i.e. if there exists a solution to the output instance of \up{} of the reduction $(T,S,h,d_T,d_S)$, then there is a solution to the \JISP{3} instance $J$. 
Consider the partial instance of \JISP{3} in \cref{fig:jisp-ab}. Note that it does not have a solution. Applying a reduction similar to the one defined above but without \abs{}, results in the \up{} instance $(T,S,h,d_T,d_S)$ in \cref{fig:up-ab}. The mapping lines in bold in \cref{fig:up-ab} are a solution for that instance.

\begin{figure}[t]
    \centering
    \begin{subfigure}[b]{0.7\textwidth}
        \centering
        \includestandalone[width=0.8\textwidth]{figures/JISP-ab-important}
        \caption{Partial \JISP{3} Instance}
        \label{fig:jisp-ab}
    \end{subfigure}
    \par\bigskip
    \begin{subfigure}[b]{0.5\textwidth}
        \centering
        \includegraphics[width=\textwidth]{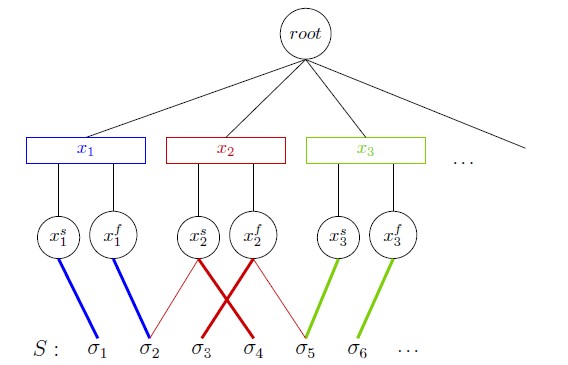}
        \caption{Partial \up{} Instance}
        \label{fig:up-ab}
    \end{subfigure}
    \caption{An example instance resulting from a reduction without the \abs{}.}
    \label{fig:ab-important}
\end{figure}

This contradiction arises because Q-node children can also be ordered from right to left. 
With \abs{} a \up{} instance is created for which only a left-to-right ordering of the children of a Q-node $x_i$ can be a part of a possible solution. 
The definition of $h$ dictates that in $\M{}$ every $x_i^a$ will be mapped to a $\sigma_a(j_i)$ and every $x_i^b$ will be mapped to a $\sigma_b(\ell_i)$. 
Because there are no deletions allowed and because of the possible reordering of the children of a Q-node, either $\ell_i = j_i+1$ (left-to-right) or $\ell_i = j_i-1$ (right-to-left). 
In $S$ the character $\sigma_a$ is always to the left of $\sigma_b$, hence there are no indices $j,\ell$ such that $\ell=j-1$, $S[j]=a$ and $S[\ell]=b$. 
So, for every $1\leq i\leq\gamma$, $\ell_i = j_i+1$. 
This means that the children of a Q-node $x_i$ are ordered form left to right as needed.

%% file: figures/reduction_interval.tex
\begin{tikzpicture}
\draw[<->] (-.5, 0) -- (6.5, 0);

\foreach \x in {0,0.75,...,5.75} {
    \pgfmathsetmacro\result{(\x / 0.75)+1}
    \pgfmathsetmacro\location{\x +0.375}
    \node[] at (\location, -.25) {\pgfmathprintnumber{\result}};
    \draw[] (\x, 0) -- (\x, -4pt);
}
\draw[] (6,0) -- (6,-4pt);

\draw[blue] (0,0) -- (0,0.2) -- (1.5,0.2) -- (1.5,0);
\node[blue] at (-0.2,0.4) {$I_1^1$};
\draw[darkred] (0.75,0) -- (.75,0.4) -- (2.25,0.4) -- (2.25,0);
\node[darkred] at (0.55,0.45) {$I_2^2$};
\draw[blue] (3,0) -- (3,0.3) -- (4.5,0.3) -- (4.5,0);
\node[blue] at (2.80,0.4) {$I_4^1$};
\draw[darkred] (3.75,0) -- (3.75,0.4) -- (5.25,0.4) -- (5.25,0);
\node[darkred] at (3.75,0.7) {$I_5^2$};
\draw[darkred] (4.5,0) -- (4.5,0.5) -- (6,0.5) -- (6,0);
\node[darkred] at (4.55,0.8) {$I_6^2$};
\draw[blue] (0.75,0) -- (.75,0.7) -- (2.25,0.7) -- (2.25,0);
\node[blue] at (.75,0.95) {$I_3^1$};

\end{tikzpicture}

%% file: figures/JISP-ab-important.tex
\begin{tikzpicture}
\draw[<->] (-1, 0) -- (7, 0);

\foreach \x in {1,...,6,7} {
    \pgfmathsetmacro\result{\x -1}
    \pgfmathsetmacro\location{(\x -0.5}
    \node[] at (\location, -.25) {\x};
    \draw[] (\result, 0) -- (\result, -4pt);
}
\node[] at (6.5, -.25) {7};

\draw[blue] (0,0) -- (0,0.3) -- (2,0.3) -- (2,0);
\draw[darkred] (1,0) -- (1,0.5) -- (3,0.5) -- (3,0);
\draw[darkred] (3,0) -- (3,0.7) -- (5,0.7) -- (5,0);
\draw[applegreen] (4,0) -- (4,0.3) -- (6,0.3) -- (6,0);
\node[blue] at (-0.2,0.4) {$I_1^1$};
\node[darkred] at (0.8,0.6) {$I_2^2$};
\node[darkred] at (2.8,0.8) {$I_3^2$};
\node[applegreen] at (3.8,0.4) {$I_4^3$};

\node at (6.5,0.4) {$\dots$};

\end{tikzpicture}

%% file: appendix.tex
\section{The Length of the Derived String}\label{par:len-derived-string}
Given a node $x$ and the numbers of deletions, $k_T$ and $k_S$, the length of $S'$, the string derived from $T(x)$, can be calculated. If there were no deletions, the length of $S'$ is equal to the span of $x$, because every leaf of $T(x)$ is mapped to exactly one character of $S$ (see \cref{fig:string-length-no-deletion}). Consider the case in which there is one deletion from the tree (\cref{fig:string-length-node-deletion}). Every one of the leaves in $T(x)$ is mapped to one character of $S$ except for the deleted leaf which is not mapped to any character. So, in this case the derivation is to a substring of length $\xspan{x}-1$. In general, if there are $k_T$ deletions from the tree (and none from the string), then the length of the substring derived from $T(x)$ is $\xspan{x}-k_T$.
Now, consider the case in which there is one deletion from the string (\cref{fig:string-length-string-deletion}). There are $\xspan{x}$ characters of $S$ that are mapped to the leaves of $T(x)$. One more character is a part of the derived substring, but it is not mapped to any of its leaves. So, in this case $T(x)$ is derived to a substring of length $\xspan{x}+1$. In general, if there are $k_S$ deletions from the string (and none from the tree), the length of the substring derived from $T(x)$ is $\xspan{x}+k_S$.
Thus, the definition of the length function $L(x,k_T,k_S) \doteq \xspan{x} - k_T + k_S$. 

\begin{figure}[t]
    \centering
    \begin{subfigure}[b]{0.4\textwidth}
        \centering
        \includegraphics[width=0.65\textwidth]{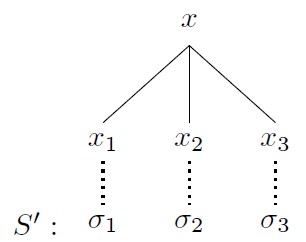}
        \caption{No deletions derives a string of length $3$ which is equal to $\xspan{x}$.}
        \label{fig:string-length-no-deletion}
    \end{subfigure}
    \hfill
    \begin{subfigure}[b]{0.4\textwidth}
        \centering
        \includegraphics[width=0.65\textwidth]{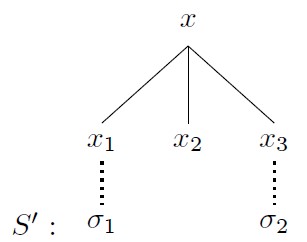}
        \caption{One deletion from the tree ($x_2$) derives a string of length $2$ which is equal to $\xspan{x} - 1$.}
        \label{fig:string-length-node-deletion}
    \end{subfigure}
    \hfill
    \begin{subfigure}[b]{0.4\textwidth}
        \centering
        \includegraphics[width=0.75\textwidth]{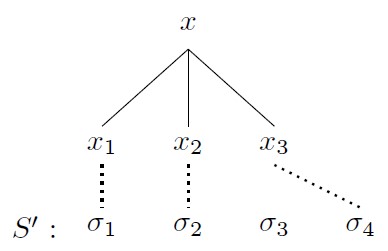}
        \caption{One deletion from the string ($\sigma_3$) derives a string of length $4$ which is equal to $\xspan{x} + 1$.}
        \label{fig:string-length-string-deletion}
    \end{subfigure}
    \hfill
    \begin{subfigure}[b]{0.4\textwidth}
        \centering
        \includegraphics[width=0.65\textwidth]{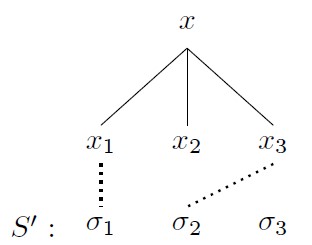}
        \caption{One deletion from the string ($\sigma_3$) and one from the tree ($x_2$) derives a string of length $3$ which is equal to $\xspan{x}$.}
        \label{fig:string-length-both-deletions}
    \end{subfigure}
    \caption{An example of the effect the number of deletion from the tree and string have on the length of the derived string. In this example the node, $x$, has a span of $3$ and the \otom{} between the children of $x$ and the characters of the string are denoted by dotted lines.}
    \label{fig:string-length}
\end{figure}

\section{The Initialization of the DP Table in the P-Mapping Algorithm}\label{sec:p-mapping-init}

The P-mapping algorithm initializes $\PP{}$ using the following two rules:
\begin{enumerate}
    \item If $L(C,k_T,k_S)=0$ and $k_S=0$, then $\PP[C,k_T,k_S]=0$.
    \item If $C=\emptyset$ and $k_T=0$, then $\PP[C,k_T,k_S]=0$.
\end{enumerate}
The first rule refers to a case where $L(C,k_T,k_S)=0$, which means that the derived substring is the empty string and thus no character can be deleted from it; hence, $k_S$ must equal $0$ (and any other value of $k_S$ is invalid). From the definition of $L(C,k_T,k_S)$, if $L(C,k_T,0)=0$, then  $k_T = \sum_{x \in C}{\xspan{x}}$, i.e. all nodes $x \in C$ are deleted. So, the score in $\PP[C,k_T,k_S]$ is $0$.
The second rule refers to a case where $C=\emptyset$, i.e. all the children of $x$ are ignored. Similarly to the first rule, a value for $k_T$ other than $0$ is invalid, and will have a $-\infty$ value. From the definition of $L$, if $k_T=0$, then $L(\emptyset,0,k_S)=k_S$, so all characters from the substring are deleted, and the score is $0$.

\section{Deleting a Child of a P-Node}\label{par:p-node-deletion-intuition}
In the P-mapping algorithm (\cref{subsec:p-mapping-alg}) it was claimed that there is no need to add to the recursion rule (\cref{eq:p-recursion}) a third case for the deletion of a child of the input node, $x$, because that case is captured in the initialization rules.
In the following example it is shown that the first initialization rule (given in \cref{sec:p-mapping-init}) is enough to enable the algorithm to find the best derivation even if it includes a node deletion, and that adding the option of deleting a node in the recursion rule is therefore redundant.
Consider the P-node $x$ in \cref{fig:p-node-deletion-intuition}, which has three leaf children ($x_1, x_2, x_3$). Assume the only derivations of the children of $x$ that have a score different than $-\infty$ are the derivation $\mu_1$ of $x_1$ to $S'[2:2]$ with no deletions, and the derivation $\mu_2$ of $x_2$ to $S'[1:1]$ with no deletions. Clearly, the best derivation of $x$ to $S'$ is the derivation that deletes $x_3$ and maps $x_1$ and $x_2$ to $S'[2]$ and $S'[1]$, respectively (denoted by dotted lines in \cref{fig:p-node-deletion-intuition}). At the end of the algorithm it is expected that this derivation can be found in the DP table entry $\PP[\{x_1,x_2,x_3\},1,0]$. Thus, let us use the recursion rule in \cref{eq:p-recursion} to compute $\PP[\{x_1,x_2,x_3\},1,0]$. The best score for $\PP[\{x_1,x_2,x_3\},1,0]$ is achieved by choosing to keep $x_1$: $\PP[\{x_1,x_2,x_3\},1,0] = \PP[\{x_2,x_3\},1,0] + \mu_1.score$. Now let us compute the best score for $\PP[\{x_2,x_3\},1,0]$. It is achieved by choosing to keep $x_2$: $\PP[\{x_2,x_3\},1,0] = \PP[\{x_3\},1,0] + \mu_2.score$. To construct the derivation $x_3$ needs to be deleted. This deletion adds $0$ to the score, and indeed, from the first initialization rule $\PP[\{x_3\},1,0] = 0$. Note that at this point it is possible to delete $x_3$ because $\xspan{x_3}=1=k_T$. Thus, we receive the score of the computed derivation is $\mu_1.score + \mu_2.score + 0$, as required.
\begin{figure}
    \centering
    \includegraphics[width=0.2\textwidth]{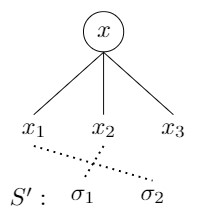}
    \caption{A P-node $x$ with three leaf children $x_1$, $x_2$ and $x_3$. The only derivations of the children of $x$ that have a score different than $-\infty$ are depicted in dashed lines.}
    \label{fig:p-node-deletion-intuition}
\end{figure}

%% file: naive.tex
\section{A Comparison with the Na\"ive Solution}\label{sec:naive}
In this section a na\"ive, alternative, algorithm for the \up{} problem is described and its time complexity is analyzed. It is shown that the time complexity of our algorithm is substantially smaller than that of the na\"ive algorithm.

Solving the \up{} problem requires a search for a \otom{} that yields a derivation of a \pqt{} $T$ to a substring of the input string $S$. That is, a substring $S'$ of $S$, such that the deletion of up to $d_S$ characters from $S'$ and the substitution of some of its characters yields a new string $S''\in C_{d_T}(T)$ (see \cref{def:C(T)}). Hence, a na\"ive way to solve the problem is to go over every string in $C_{d_T}(T)$ and try to find an alignment between it and every substring of $S$, when only $d_S$ deletions are allowed from $S$. Equivalently, it is possible to search for an alignment between every substring of $S$ and every string $S_T \in C_0(T)$ with up to $d_S$ deletions from $S$ and up to $d_T$ deletions from $S_T$. 

Naturally, sequence alignment can be used, but in order to bound the number of deletions, the basic algorithm needs to be modified. The usual $2$-dimensional DP-table needs to be extended with two additional dimensions that correspond to the numbers of deletions from $S$ and $S_T$. This way, when filling the table, the best scoring alignment considered so far for every deletion numbers combination can be stored. At the end of the algorithm, the score of the best alignment is the maximum between the entries of the DP-table corresponding to an alignment between $S_T$ and a prefix of $S$ that has a length between $m-d_T$ and $m+d_S$.
Thus, the outline of the na\"ive algorithm is as follows. For every string $S_T \in C(T)$ and every possible start index $i$, preform sequence alignment with a bounded number of deletions. Then, find the index $i$ that resulted in the highest scoring alignment.

The size of the DP table is $O(m(m+d_S)d_T d_S)$, but in the first two dimensions only a diagonal with a width of $O(d_T + d_S)$ entries needs to be computed. The computation of each entry takes $O(1)$ time and finding the best alignment takes $O(d_T d_S)$. Thus, every run of the sequence alignment with a bounded number of deletions and a specific start index $i$ takes $O(m(d_T+d_S)d_T d_S)$ time. As seen in \cref{sec:general-algorithm}, there are $O(n)$ possible values for $i$.

Finally, let us bound the number of strings in $C(T)$ which is equal to the number of \pqt{s} that are equivalent to $T$.
By definition, every legal permutation of the children of an internal node of $T$ results in a new \pqt{} $T'$ that is equivalent to $T$, i.e. $T \equiv T'$ (equivalence and not quasi-equivalence is used here because $C_0(T)$ is considered, i.e. there are no deletions from the tree).
The children of a Q-node can be permuted only in one of two ways (left-to-right or right-to-left) and the children of a P-node can be arranged in any order. So, for an internal node $x$ of $T$, for every string resulting from the rearrangement of the children of all the other nodes in $T$, $x$ contributes $2$ strings to $C(T)$ if it is a Q-node, and ${\gamma}!$ strings if it is a P-node.
Thus, $|C(T)|= O(2^{m_q}{(\gamma!)}^{m_p})$.
In total, the na\"ive solution for \up{} takes $O(2^{m_q}{(\gamma!)}^{m_p} n m(d_T+d_S)d_T d_S)$ time.

Both algorithms have a factor of $O(n d_T d_S)$, so it can be ignored and more concise time complexities can be compared: the na\"ive $O(2^{m_q}{(\gamma!)}^{m_p} m(d_T+d_S))$ versus our $O(d_T d_S (m_p \gamma 2^\gamma + m_q \gamma))$.
In both algorithms the non-polynomial factors in the time complexity are dependent on the number of P-nodes and the number of Q-nodes, so let us consider two complementary cases. 
First, assume there are only P-nodes in the \pqt{} (i.e. $m=m_p$). In this case, the na\"ive algorithm has a ${(\gamma!)}^{m_p}={[2^{O(\gamma \log\gamma)}]}^{m}$ factor, which is super-exponential in $\gamma$, and even worse, exponential in $m$, while our algorithm has a $m_p \gamma 2^\gamma=m \gamma 2^\gamma$ factor which is exponential only in $\gamma$, and in particular polynomial for any $\gamma$ that is constant (or even logarithmic in the input size). 
Second, assume there are only Q-nodes in the \pqt{} (i.e. $m=m_q$). In this case, the na\"ive algorithm has a $2^{m_q}=2^m$ factor, which is exponential and our algorithm has a $m_q \gamma = m \gamma$ factor, which is polynomial.

%% file: algorithm/Q-node.tex
\section{Q-Node Mapping} \label{sec:q-node-mapping}
In this section we describe the Q-mapping algorithm called by the main algorithm described in \cref{sec:general-algorithm}
\subparagraph{Objective.} \label{par:q-input-output} 
As already mentioned in \cref{sec:algorithm}, the Q-mapping algorithm receives the following as input.
\begin{enumerate}
    \item An internal node $x$ that is a Q-node and has $\gamma$ children: $x_1,\dots,x_\gamma$.
    \item A string $S'$ (which is a substring of the original $S$).
    \item A collection of derivations $\D{}$ of the children of $x$ to substrings of $S'$. 
    The derivations are grouped by their root nodes $\mu.v$, and ordered by their end points, $\mu .e$.
    \item The maximum number of deletions from the tree and string, $d_T$ and $d_S$, respectively.
\end{enumerate}

The output of the algorithm is the set $\bigcup_{k_T \leq d_T}\bigcup_{k_S \leq d_S}\max_{\mu \in \Mx{}}\mu.score$, which is a set of derivations of $x$ to prefixes of $S'$. The set holds the best scoring derivation for every possible deletion number combination $k_T,k_S$ where $0 \leq k_T \leq d_T$ and $0 \leq k_S \leq d_S$. The set is ordered by the end points of the derivations and it is of size $O(d_T d_S)$.
Note that the input and output of this algorithm is the same as the input and output of the P-mapping algorithm (\cref{subsec:p-mapping-alg}), except for the type of the node received as input.

As a start, a solution assuming that the children of the Q-node $x$ can only be arranged in a left-to-right order is demonstrated. The fact that they can also be arranged in a right-to-left order is addressed at the end of this section.
The Q-mapping algorithm is a DP algorithm that uses a 3-dimensional DP table, $\Q{}$. The pseudocode of our algorithm can be found in \cref{alg:q-node} ahead.

\subparagraph{Notations.}
At different stages of the algorithm the node $x$ is considered as if it has only its first $i$ children.
For an index $i$ such that $1\leq i\leq \gamma$ (namely, $i$ is an index of a child of $x$), two definitions are given. The first, $\xis{}$, denotes the set of the first $i$ children of $x$. Formally, $\xis{} \doteq \{x_1,\dots ,x_i\}$. The second, $\x{i}$, denotes the node $x$ considering only its children in $\xis{}$. Consequentially, the span of $\x{i}$ is defined as $\sum_{j=1}^i{\xspan{x_j}}$ and the set $\Mxi{}$ (in \cref{def:M} where $U=\{\x{i}\}$) now refers to a set of {\em partial} derivations.
To use $\x{i}$ to describe the base cases of our algorithm, let us define $\x{0}$ ($\x{i}$ for $i=0$) as a tree with no labeled leaves to map.

\subparagraph{The DP Table.}
The purpose of an entry $\Q[i,k_T,k_S]$ is to hold the score of the best partial derivation of $\x{i}$ to a prefix of $S'$ with exactly $k_T$ deletions from the tree and exactly $k_S$ deletions from the string. Namely, only the first $i$ children of $x$, $x_1,\dots,x_i$, are considered and the rest are ignored. The other children of $x$, that is $\children{x}\setminus\xis{}$, are accounted for in the computation of other entries of $\Q{}$. Formally, $\Q[i,k_T,k_S] = \max_{\mu \in \Mxi{}}\mu.score$.

Similarly to the main algorithm (\cref{sec:general-algorithm}) and the P-mapping algorithm (\cref{subsec:p-mapping-alg}), some of the entries of the DP table are invalid, and their value is defined as $-\infty$ throughout the algorithm. 
Here we give a more detailed description of these entries for all three algorithms and their DP tables.
For a given DP table, the invalid entries are the ones that their indices define an illegal derivation. Namely, derivations that have more deletions from the tree than there are leaves in the subtree of $T$ rooted in the derived node, derivations that have more deletions from the string than there are characters in the derived string, derivations that derive a string that by definition ends in an index larger than the end index of the input string, or derivations that by definition derive a string with a negative length. 
Thus, an entry $\Q[i,k_T,k_S]$ is invalid if one of the following is true:
$k_T > \sum_{c \in \xis{}}\xspan{c}$, $k_S > L(\xis{},k_T,k_S)$, $L(\xis{},k_T,k_S) > \len{S'}$, or $L(\xis{},k_T,k_S) < 0$.
Similarly, an entry $\PP[C,k_T,k_S]$ is invalid if one of the following is true: $k_T > \sum_{c \in C}\xspan{c}$, $k_S > L(C,k_T,k_S)$, $L(C,k_T,k_S) > \len{S'}$, or $L(C,k_T,k_S) < 0$. 
Lastly, an entry $\A[j,i,k_T,k_S]$ is invalid if one of the following is true: $k_T > \xspan{x_j}$, $k_S > L(j,i,k_T,k_S)$, $E(j,i,k_S,k_T) > n$, or $L(j,i,k_T,k_S) < 0$. 

\subparagraph{Filling the DP Table.} The algorithm initializes $\Q{}$ as follows. For every $0 \leq k_S \leq \len{S'}$, $\Q[0,0,k_S]=0$. These entries of the DP table capture the cases in which a prefix of $S'$ of length $L(\emptyset,0,k_S) = k_S$ is derived, i.e. there are no leaves to map. Thus, all the characters in $S'[1:k_S]$ must be deleted under this partial derivation. This is possible because the allowed number of deletions from the string is exactly the number of characters in the derived substring. Note that $k_S \leq \len{S'}$ because otherwise $\Q[0,0,k_S]$ is an invalid entry and its value should remain $-\infty$.

Afterwards, the remaining entries of $\Q{}$ are calculated using the recursion rule in \cref{eq:q-recursion} ahead. The order of computation is ascending with respect to $i$ (i.e. $i = 1,\dots,\gamma$), for a given $i$, the order of computation is ascending with respect to the number of deletions from the string (i.e. $k_S=0,1,\dots,d_S$), and for a given $i$ and $k_S$, the order of computation does not matter. Nonetheless, an ascending order with respect to the number of deletions from the tree (i.e. $k_T=0,1,\dots,d_T$) is set.
\begin{equation} \label{eq:q-recursion}
    \Q[i,k_T,k_S] = 
            \max \begin{cases}
                    \Q[i,k_T,k_S-1] \\
                    \Q[i-1,k_T-\xspan{x_i},k_S] \\
                    \displaystyle\max_{\substack {\mu \in \Mi{} \\ \mathrm{s.t.}\ \mu .v = x_i}} \Q[i-1, k_T-\mu .\dt,k_S-\mu .\ds] + \mu .score\\
                \end{cases}
\end{equation}
The intuition behind \cref{eq:q-recursion} is that given a partial derivation, $\mu \in \Mxi$, one of the three cases of the rule must be true. The end point of $\mu$ is $\E(\xis{},k_T,k_S)$, and thus, by definition, $S'[\E(\xis{},k_T,k_S)]$ is either deleted under $\mu$ (the first case) or it is mapped under $\mu$ (the third case). The partial derivation $\mu$ does not ignore $x_i$, so either it is kept under $\mu$ (the third case), or it is deleted under $\mu$ (the second case).

In the first case, $S'[\E(\xis{},k_T,k_S)]$ is deleted under $\mu$. Removing the deletion of $S'[\E(\xis{},k_T,k_S)]$ from $\mu$ (formally defined in \cref{def:rm-char-del}) results in a partial derivation, $\mu'$, that considers the same subset of children of $x$ with the same number of deletions from the tree and one less deletion from the string. That is, $\mu' \in \D(\x{i},k_T,k_S-1)$, and the score of the best partial derivation with the same properties as $\mu'$ is in $\Q[i,k_T,k_S-1]$.
In the second case, $x_i$ is deleted under $\mu$. Removing the deletion of $x_i$ from $\mu$ (formally defined in \cref{def:rm-node-del}) results in a partial derivation of $\x{i-1}$ with $k_T-\xspan{x_i}$ deletions from the tree. That is, a derivation in $\D(\xis{-1},k_T-\xspan{x_i},k_S)$, and the score of the best one is in $\Q[i-1,k_T-\xspan{x_i},k_S]$.
In the third case there is a derivation, $\mu'$, of one of the children of $\x{i}$ such that $S'[\E(\xis{},k_T,k_S)]$ is mapped under it. Because $x_i$ is kept under $\mu$ (in derivation this entry holds) and it is the last child of $\x{i}$, then $\mu'$ is a derivation of $x_i$ (i.e. $\mu'.v = x_i$). Otherwise, the ordering of the children of the Q-node $x$ is illegal. The score of $\mu$ in this case is equal to the score of $\mu'$ plus the score of a partial derivation of $\x{i-1}$ with $k_T-\mu' .\dt$ and $k_S-\mu' .\ds$ deletions from the tree and string, respectively. The best score of a partial derivation in $\D(\x{i-1},k_T-\mu' .\dt,k_S-\mu' .\ds)$ is in $\Q[i-1, k_T-\mu .\dt,k_S-\mu .\ds]$. Thus we try to find the derivation, $\mu'$, of $x_i$ with the end point $\E(\xis{},k_T,k_S)$ that maximises $\mu'.score + \Q[i-1, k_T-\mu' .\dt,k_S-\mu' .\ds]$.

\subparagraph{Finding the Solution.}
Our goal is to find a derivation of $x$ according to all of its children, hence once the entire DP table is filled our solution is in $\Q[\gamma,\cdot,\cdot]$. The best derivation for every deletion combination should be ordered with respect to the end point of the derivation. Here we explain how this can be done by simply traversing $\Q{}$ in a predefined order and without any further calculation.
First, note that there could be more than one derivation per end point. For example, the second smallest end point (the end point of the second shortest substring derived from $x$) is generated by the deletion combination $(d_T,1)$, thereby $e = \E(\gamma,d_T,1) = \sum_{k=1}^\gamma \xspan{x_k} -d_T + 1 = \xspan{x} - d_T + 1$. The deletion combination $(d_T-1,0)$ also yields $e$, $\E(\gamma,d_T-1,0) = \sum_{k=1}^\gamma \xspan{x_k} -(d_T - 1) + 0 = \xspan{x} - d_T + 1 = e$.
In fact, only the smallest and largest end points have just one derivation each. The deletion combination $(d_T,0)$ yields the smallest end point and $(0,d_S)$ yields the largest.
Thus, given the $(d_T+1)\times (d_S+1)$ sized table $\Q[\gamma,\cdot,\cdot]$, the ordered list of derivations can be generated by traversing the table in the order specified in \cref{tab:deletion-traversal-order}.

\begin{table}[t]
    \centering
    \begin{tabular}{|c|*{3}{c|}}\hline
        \backslashbox[3em]{$k_T$}{$k_S$}
        &\makebox[3em]{$0$}&\makebox[3em]{$1$}&\makebox[3em]{$2$}\\\hline
        $0$ & $e_4$ & $e_5$ & $e_6$ \\\hline
        $1$ & $e_3$ & $e_4$ & $e_5$ \\\hline
        $2$ & $e_2$ & $e_3$ & $e_4$ \\\hline
        $3$ &$e_1$ & $e_2$ & $e_3$ \\\hline
    \end{tabular}
    \caption{An example of the ordering between the end points induced by the different deletion combinations. In this example, $d_T = 3$, $d_S = 2$ and the end points are indexed from first to last ($e_i = 1 + e_{i-1}$ for $2\leq i \leq 6$).}
    \label{tab:deletion-traversal-order}
\end{table}

\begin{algorithm}[t]
\SetAlgoLined
\SetKwComment{Comment}{//}{}
\KwIn{$x, S', \D{}, d_T, d_S$}
\KwOut{The best derivation of $x$ to every prefix of $S'$}
$\gamma \gets |\children{x}|$\;
build $\Q{}$ with dimensions $\gamma+1 \times d_T+1 \times d_S+1$\;
\For{$k_S=0$ \KwTo $\len{S'}$}
{\Comment{initialization}
    $\Q[0,0,k_S] \gets 0$\;
}
\For{$i=1$ \KwTo $\gamma$}{
    \For{$k_S=0$ \KwTo $d_S$}{
        \For{$k_T=0$ \KwTo $d_T$}{
            compute $\Q[i,k_T,k_S]$ according to \cref{eq:q-recursion} \;
        }
    }
 }
 \KwRet{$\Q[\gamma,\cdot,\cdot]$} \;
\caption{Q-Mapping}
\label{alg:q-node}
\end{algorithm}

\subparagraph{A Second Ordering of the Children.}\label{par:q-node-second-order} Previously an algorithm to find a \otom{} for a tree rooted in a Q-node assuming its children can only be arranged in a left-to-right order was described. To consider also a right-to-left arrangement, the following minor modification to the algorithm is required. Run the first two parts of the algorithm described above twice, each run fills a different DP table. The first run of the algorithm will receive the children of $x$ from left to right (i.e. $x_1$ is the leftmost child of $x$ and $x_\gamma$ is the rightmost child), and will produce a DP table, $\Q_\ell$, holding the best scores of the partial derivations of $x$ that order its children from left to right. The second run will receive the children of $x$ from right-to-left (i.e. $x_1$ is the rightmost child of $x$ and $x_\gamma$ is the leftmost child), and will produce a DP table, $\Q_r$, holding the best scores of the partial derivations of $x$ that order its children from right to left.
To find the solution, go over both DP tables as described above (\cref{tab:deletion-traversal-order}), but for every deletion combination $k_T,k_S$, return the maximum between $\Q_\ell[\gamma,kT,kS]$ and $\Q_r[\gamma,kT,kS]$.

\subparagraph{Time and Space Complexity.}
The DP table is the most space consuming data structure in the described algorithm. Its dimensions are $\gamma+1 \times d_T+1 \times d_S+1$, and the algorithm uses two of them. 
The computation of an entry of the DP table, $\Q[i,k_T,k_S]$, includes two $O(1)$ calculations (the first and second cases of the recursion rule) and going over every derivation of $x_i$ in $\Mi{}$.
All those derivations have the same root and the same end point, but a different number of deletions. In fact, there are no two derivations of $x_i$ in $\Mi{}$ that have the same deletion combination ($k'_T,k'_S$). Hence, the number of such derivations is equal to the number of deletion combinations, $k_T\cdot k_S$, and so the calculation of an entry of the DP table takes $O(k_T k_S) = O(d_T d_S)$ time.
Thus, the time complexity of the algorithm is $O(\gamma {d_T}^2 {d_S}^2)$.

In the previous paragraph the calculation of $\E{}$ for every entry, which yields the relevant end point for the entry, was ignored. The most time consuming part of that calculation is the summation of spans ($\sum_{k=1}^i \xspan{x_k}$) which takes $O(\gamma)$ time. To prevent the wasteful repetition, these summations are calculated once and then saved in a table of size $\gamma$. These summations are calculated twice - once for each possible children ordering (left-to-right and right-to-left). This is negligible with respect to the time it takes to fill the DP table.

%% file: del-penalty.tex
\section{Penalizing Deletions}\label{sec:del-penalty}
To assign deletions a penalization cost (and not only limit them), the algorithm should receive as input a deletion penalty function, $\delta:\Sigma_T \cup \Sigma_S \rightarrow \R{}$. The function defines the penalty of deleting a character from $S$ or a leaf from $T$ according to its label. Then, let us expand $\delta$, and define the deletion penalty of a node $x$ in $T$ as the summation of the deletion penalty of all the leaves in the subtree rooted in $x$. Thus, the set of nodes in $T$ is denoted by $T.nodes$, and a new function $\Delta: T.nodes \cup \Sigma_S \rightarrow \R{}$ is defined in \cref{eq:Delta} below. Note that the $\Delta$ function can be calculated in advance, by going over $T$ in postorder. This calculation takes $O(m') = O(m)$ time.
\begin{equation} \label{eq:Delta}
    \Delta(x) = 
        \begin{cases}
            \delta(x), & if \ x \in \Sigma_S \\
            \displaystyle\sum_{\ell \in \leaves{x}}{\delta(\mathsf{label}(\ell))}, & if \ x \in T.nodes
        \end{cases}
\end{equation}
In addition, the following changes to the main algorithm and to the P-mapping and Q-mapping algorithms are needed.
First, the initialization of the main DP table $\A{}$ should change and add to the score of every leaf entry (i.e. $\A[j,i,k_T,k_S]$ such that $x_j$ is a leaf) the cost of the deleted nodes and characters. Namely, in \cref{alg:up} (given in \cref{sec:main-alg-proof}) lines 8-9 should be replaced with \cref{eq:alg-init-penalty} below. 
Second, the $\Delta$ function in \cref{eq:Delta} should be sent from the main algorithm to the Q-mapping and P-mapping algorithms.
\begin{equation} \label{eq:alg-init-penalty}
\begin{split}
	&\A[j,i,1,k_S] \gets -\Delta(x_j) - \sum_{\ell=i}^{i+k_S-1}{\Delta(S[\ell])}\\
	&\A[j,i,0,k_S] \gets \max_{i'=i,...,i+d_S}{h(j,i') - \sum_{\ell=i}^{i+k_S}{\Delta(S[\ell])} + \Delta(S[i'])}
\end{split}
\end{equation}

Third, the initialization of the DP table, $\PP{}$, and the recursion rule of the P-mapping algorithm need to change.
The first initialization rule, where $L(C,k_T,k_S)=0$ and $k_S=0$, depicts the case in which every node in $C$ is deleted, hence, the rule should be $\PP[C,k_T,k_S] = -\sum_{c\in C}{\Delta(c)}$. The second rule, where $C=\emptyset$ and $k_T=0$, concerns the case in which every character in $S'[1:k_S]$ is deleted, so the rule should be $\PP[C,k_T,k_S] = -\sum_{i=1}^{k_S}{\Delta(S'[i])}$.
The recursion rule should be changed to the one in \cref{eq:p-recursion-penalized}. Note that the change is only in the first case where the cost of deleting the $i^{\textrm{th}}$ character of $S'$ is subtracted from the score.

\begin{equation}\label{eq:p-recursion-penalized}
    \PP[C,k_T,k_S] = 
        \max \begin{cases}
            \PP[C,k_T,k_S-1] - \Delta(S[i]) \\
            \displaystyle\max_{\mu \in \Mc{}} \PP[C\setminus \{\mu.v\},k_T-\mu.\dt,k_S-\mu.\ds] + \mu.score\\
        \end{cases}
\end{equation}

Lastly, in the Q-mapping algorithm the initialization of the DP table, $\Q{}$, and the recursion rule also need to change. In the initialization, for every $0\leq k_S\leq d_S$, $\Q[0,0,k_S] = -\sum_{i=0}^{k_S}\Delta(S'[i])$. That is because this is the case in which every character in $S'[1:k_S]$ is deleted. When filling the DP table, the recursion rule in \cref{eq:q-recursion-penalty} should be used. Note the change is in the first and second cases. In the first the cost of deleting the $i^{\textrm{th}}$ character of $S'$ is subtracted from the score, and in the second the score for deleting $x_i$ is subtracted from the score.

\begin{equation} \label{eq:q-recursion-penalty}
    \Q[i,k_T,k_S] = 
            \max \begin{cases}
                    \Q[i,k_T,k_S-1] - \Delta(S[\E(i,k_T,k_S)])\\
                    \Q[i-1,k_T-\xspan{x_i},k_S] - \Delta(x_i) \\
                    \displaystyle\max_{\substack {\mu \in \Mi{} \\ \mathrm{s.t.}\ \mu .v = x_i}} \Q[i-1, k_T-\mu .\dt,k_S-\mu .\ds] + \mu .score\\
                \end{cases}
\end{equation}

%% file: algorithm/proof/proofs.tex
\section{Correctness and Runtime Analysis of Our Algorithms}\label{sec:proofs}
In this section we prove the correctness of the \up{} algorithm (\cref{sec:main-alg-proof}), the P-mapping algorithm (\cref{sec:p-node-proof}) and the Q-mapping algorithm (\cref{sec:q-node-proof}), and prove the time complexity of the \up{} algorithm and the P-mapping algorithm. First, some definitions that are used in the proofs are given.

\subparagraph{Addition and Removal of a Derivation.}
Given a partial derivation, $\mu$, which derives an internal node, $x$, let us define the removal and addition of another derivation $\eta$: $\remove{\mu}{\eta}$ and $\add{\mu}{\eta}$. Both operations are defined only for a derivation $\eta$ whose root is a node $x' \in \children{x}$.
\begin{definition}\label{def:remove-derivation}
The operation $\remove{\mu}{\eta}$ is defined only if $\eta$ is the derivation of $\eta.v$ under $\mu$ and if at least one among $\eta.e = \mu.e$ or $\eta.s=\mu.s$ is true. 
The operation returns a new partial derivation $\mu'$ of $\mu.v$ that ignores the subtree of $T$ rooted in the child node $\eta.v$. If $\eta.e = \mu.e$, then $\mu'$ derives the string $S[\mu.s:\eta.s-1]$, and if $\eta.s=\mu.s$, then $\mu'$ derives the string $S[\eta.e+1:\mu.e]$. In any case the number of deletions from the tree is $\mu'.\dt = \mu.\dt - \eta.\dt$ and from the string it is $\mu'.\ds = \mu.\ds - \eta.\ds$.
Furthermore, $\mu.o \setminus \eta.o$ is the \otom{} that yields $\mu'$.
\end{definition}
\begin{definition}\label{def:add-derivation}
The operation $\add{\mu}{\eta}$ is defined only if either $\eta.s = \mu.e+1$ or $\eta.e=\mu.s-1$ is true and if the node $\eta.v$ is ignored under $\mu$.
The operation returns a new partial derivation $\mu'$ of $\mu.v$. The derivation of $\eta.v$ under $\mu'$ is $\eta$, and the mapping or deletion of every other leaf or character in the string is defined the same as it was in $\mu$. Consequentially, if $\eta.s = \mu.e+1$, then $\mu'$ derives the string $S[\mu.s:\eta.e]$, and if $\eta.e=\mu.s-1$, then $\mu'$ derives the string $S[\eta.s:\mu.e]$. Furthermore, $\mu'.\dt = \mu.\dt + \eta.\dt$, $\mu'.\ds = \mu.\ds + \eta.\ds$ and the \otom{} that yields $\mu'$ is $\mu.o \cup \eta.o$.
\end{definition}

\subparagraph{Addition and Removal of a Deleted Character.}
Given a partial derivation $\mu$, which derives a string $S$, and an index $i$ of $S$ let us define the removal and addition of a deleted character: $\rmdel{\mu}{i}$ and $\adddel{\mu}{i}$.
\begin{definition} \label{def:rm-char-del}
The operation $\rmdel{\mu}{i}$ is defined only if $i=\mu.e$ or $i=\mu.s$, and if $S[i]$ is deleted under $\mu$. The operation returns a partial derivation $\mu'$ with $\mu.\ds-1$ deletions from the string. If $i=\mu.e$, then $\mu'$ derives the string $S[\mu.s,\mu.e-1]$, and if $i=\mu.s$, then $\mu'$ derives the string $S[\mu.s+1,\mu.e]$.
The \otom{} that yields $\mu'$ is $\mu.o \setminus \{(\varepsilon , S[i](i))\}$.
\end{definition}
\begin{definition}\label{def:add-char-del}
The operation $\adddel{\mu}{i}$ is defined only if $i=\mu.e+1$ or $i=\mu.s-1$. The operation returns a partial derivation $\mu'$ with $\mu.\ds+1$ deletions from the string. If $i=\mu.e+1$, then $\mu'$ derives the string $S[\mu.s,\mu.e+1]$, and if $i=\mu.s-1$, then $\mu'$ derives the string $S[\mu.s-1,\mu.e]$.
The \otom{} that yields $\mu'$ is $\mu.o \cup \{(\varepsilon , S[i](i))\}$.
\end{definition}

\subparagraph{Addition and Removal of a Deleted Node.}
Given a partial derivation $\mu$, which derives a string $S$, let us define the removal and addition of a deleted node: $\rmdel{\mu}{x}$ and $\adddel{\mu}{x}$.
\begin{definition}\label{def:rm-node-del}
The operation $\rmdel{\mu}{x}$ is defined only if $x$ is deleted under $\mu$. The operation returns a partial derivation $\mu'$ with $\mu.\dt-\xspan{x}$ deletions from the tree and $\mu.\ds$ deletions from the string. The derivation $\mu'$ ignores $x$ and derives the same substring derived by $\mu$.
The \otom{} that yields $\mu'$ is $\mu.o \setminus \{(\ell, \varepsilon) : \ell \in \leaves{x}\}$.
\end{definition}
\begin{definition}\label{def:add-node-del}
The operation $\adddel{\mu}{x}$ is defined only if $x$ is ignored under $\mu$. The operation returns a partial derivation $\mu'$ with $\mu.\dt+\xspan{x}$ deletions from the tree. The substring derived by $\mu'$ is equal to the substring derived by $\mu$.
The \otom{} that yields $\mu'$ is $\mu.o \cup \{(\ell, \varepsilon) : \ell \in \leaves{x}\}$.
\end{definition}
\input{algorithm/proof/main-proof}
\input{algorithm/proof/p-proof}
\input{algorithm/proof/q-proof}
\subsection{Time and Space Complexity of the \PQT{} Search Algorithm}
Here we prove \cref{lemma:up-time}.
\begin{proof}\label{proof:general-time}
The number of leaves in the PQ-tree $T$ is $m$, hence there are $O(m)$ nodes in the tree, i.e the size of the first dimension of the DP table, $\A{}$, is $O(m)$. In the algorithm description a bound for the possible start indices of substrings derived from nodes in $T$ is given. The node with the largest span in $T$ is the root which has a span of $m$. The root is mapped to the longest substring when there are $d_S$ deletions from the string. Hence, the size of the second dimension of $\A{}$ is $\Omega(n-(m+d_S)+1) = \Omega(n)$ (given that $d < m << n$). The nodes with the smallest spans are the leaves, which have a span of $1$, hence the size of the second dimension of $\A{}$ is $O(n)$. The third and fourth dimensions of $\A{}$ are of size $d_T+1$ and $d_S+1$, respectively. In total, the DP table $\A{}$ is of size $O(d_T d_S m n)$.

In the initialization step $O(d_T d_S m n)$ entries of $\A{}$ are computed in $O(1)$ time each. This holds because there are $m$ leaves and $n$ possible start indices for strings of length $1$. The $d_T$ and $d_S$ factors come from the initialization of entries with $-\infty$.
The P-mapping algorithm is called for every P-node in $T$ and every possible start index $i$, i.e. the P-mapping algorithm is called $O(n m_p)$ times. Similarly, the Q-mapping algorithm is called $O(n m_q)$ times. Thus, it takes $O(n\ (m_p \ \cdot \textrm{Time(P-mapping)} + \ m_q \ \cdot \textrm{Time(Q-mapping)})))$ time to fill the DP table.
In the final stage of the algorithm (line $21$ in \cref{alg:up}) the maximum over the entries corresponding to every combination of deletion number and start index ($0\leq k_T\leq d_T$, $0\leq k_S \leq d_S, 1\leq i\leq n-(\xspan{x}-d_T)+1\}$) is computed. So, it takes $O(d_T d_S n)$ time to find the maximum score of a derivation. Tracing back through the DP table to find the actual mapping a does not increase the time complexity.

In \cref{sec:q-node-mapping} it is shown that our Q-mapping algorithm takes $O(\gamma {d_T}^2 {d_S}^2)$ time and $O(d_T d_S \gamma)$ space. From \cref{lemma:p-time} our P-mapping algorithm takes $O(\gamma 2^{\gamma} {d_T}^2 {d_S}^2)$ time and $O(d_T d_S 2^\gamma)$ space. Thus, in total, our algorithm runs in $O(n (m_p \cdot \gamma 2^{\gamma} {d_T}^2 {d_S}^2 + m_q \cdot \gamma {d_T}^2 {d_S}^2)) = O(n \gamma {d_T}^2 {d_S}^2 (m_p \cdot 2^{\gamma} + m_q))$ time. Adding to the space required for the main DP table the space required for the P-mapping algorithm (the space needed for the Q-mapping algorithm is insignificant with respect to the P-mapping algorithm) results in a total space complexity of $O(d_T d_S m n) + O(d_T d_S 2^\gamma) = O(d_T d_S (m n + 2^\gamma))$.
This completes the proof.
\end{proof}

\subsection{Time and Space Complexity of the P-Mapping Algorithm}
Here we prove \cref{lemma:p-time} below.
\begin{lemma}\label{lemma:p-time}
The P-mapping algorithm takes $O({d_T}^2 {d_S}^2 \gamma 2^\gamma)$ time and $O(d_T d_S 2^\gamma)$ space.
\end{lemma}
\begin{proof}\label{proof:p-time} 
The most space consuming part of the algorithm is the 3-dimensional DP table. The first dimension, $C$, can be any subset of the set $\children{x}$, and therefore it is of size $2^{|\children{x}|} = 2^\gamma$. The sizes of the second and third dimensions (i.e. $k_T$ and $k_S$) are bounded by $d_T+1$ and $d_S+1$, respectively. Hence, the space of the DP algorithm is $O(d_T d_S 2^\gamma)$.

The algorithm has three parts: initialization, filling the DP table, and constructing the solution. 
The most time consuming calculation required in the initialization is the calculation of $L(C,k_T,k_S)$ in the first rule. It requires summing the spans of all nodes in $C$. This calculation will also be required in the second part of the algorithm. To avoid the repetitive calculations, it preformed once for every $(C,k_T,k_S)$ tuple and save the results. This requires $O(d_T d_S 2^{|\children{x}|}) = O(d_t d_S 2^\gamma)$ space (for this is the number of such tuples). Each value is calculated in $O(|\children{x}|) = O(\gamma)$ time. Hence, the calculation of all the $L(C,k_T,k_S)$ values (and thus all the $\E(C,k_T,k_S)$ values) takes $O(d_T d_S \gamma 2^\gamma)$ time and $O(d_T d_S 2^\gamma)$ space.
The second step is done by calculating the value of every entry in the $O(d_T d_S 2^\gamma)$ entries of $\PP{}$, using the recursion rule in \cref{eq:p-recursion}. The first line among the rule takes $O(1)$ time, since it involves looking in another entry of $\PP{}$ and basic computations. The second line of the rule involves going over all derivations $\mu \in \Mc{}$. Namely, going over all derivations with a specific end point, which derives a node in $C$ and has no more than a specific number of deletions from the tree and string (i.e. $\mu.e=\E(C,k_T,k_S)$, $\mu.v \in C$, $\mu.\dt\leq k_T$ and $\mu.\ds\leq k_S$). The number of deletions from the tree and string are bounded by $d_T$ and $d_S$, respectively, and the number of nodes in $C$ is bounded by the number of children of $x$, $\gamma$.  Hence, the time to calculate one entry of $\PP{}$ is $O(d_T d_S \gamma)$. In total, the second part of the algorithm takes $O({d_T}^2 {d_S}^2 \gamma 2^\gamma)$ time.
Finally, to construct the solution the algorithm goes over every deletion combination $k_T,k_S$ once, i.e. it takes $O(d_T d_S)$ time.
In total, the algorithm takes $O({d_T}^2 {d_S}^2 \gamma 2^\gamma) + O(d_T d_S \gamma 2^\gamma) + O(d_T d_S) = O({d_T}^2 {d_S}^2 \gamma 2^\gamma)$.
\end{proof}

%% file: algorithm/proof/main-proof.tex
\subsection{The Main Algorithm}\label{sec:main-alg-proof}
In this section we give the pseudocode of the \up{} algorithm presented in \cref{sec:general-algorithm} (\cref{alg:up}) and prove its correctness. 
In this proof, the correctness of the Q-mapping algorithm (\cref{sec:q-node-mapping}) and of the P-mapping algorithm (\cref{subsec:p-mapping-alg}) is assumed. Their correctness will be proven in \cref{sec:q-node-proof} and \cref{sec:p-node-proof}, respectively.

\begin{algorithm}[t]
\SetAlgoLined
\SetKwComment{Comment}{//}{}
\KwIn{$T, S, h, d_T, d_S$}
\KwOut{The score of the best derivation of $T$ to a substring of $S$ with up to $d_T$ and $d_S$ deletions from $T$ and $S$, respectively}
build $\A{}$ with dimensions $m' \times n \times d_T+1 \times d_S+1$ and initial value $-\infty$\;
\For{$j=1$ \KwTo $m'$}{
 \Comment{for each node of $T$ in postorder}
    \For{$i=1$ \KwTo $n$}{
        \If {$x_j$ is a Leaf}{
        	\Comment{initialization}
        	\For{$k_S=0$ \KwTo $d_S$}{
                $\A[j,i,1,k_S] \gets 0$\;
                   $\A[j,i,0,k_S] \gets \displaystyle{\max_{\substack{i'=i,...,i+k_S}}}h(j,S[i'])$\;
            }
        }
        $e \gets E(x_j, i, 0, d_S)$\;
        \If {$x_j$ is a Q-node} {
            $\A[j,i,\cdot,\cdot] \gets$ Q-Mapping($x_j,S[i,e], \{\A[x_{j_k},i,\cdot,\cdot] : x_{j_k}\in\children{x_j}\}, d_T, d_S$)\;
        }
        \If {$x_j$ is a P-node} {
            $\A[j,i,\cdot,\cdot] \gets$ P-Mapping($x_j, S[i,e], \{\A[x_{j_k},i,\cdot,\cdot] : x_{j_k}\in \children{x_j}\}, d_T, d_S$)\;
        }
    }
}
 \KwRet{$\displaystyle\max_{\substack{
            0 \leq k_T \leq d_T \\ 0 \leq k_S \leq d_S \\ 1 \leq i \leq (n-(\xspan{root}-d_T)+1)}} \A[m',i,k_T,k_S]$} \;
\caption{\up{}}
\label{alg:up}
\end{algorithm}

For this proof  \cref{def:M-general} below is used to represent the set of derivations whose score might be in $\A[j,i,k_T,k_S]$, similarly to the notation in \cref{def:M}.
\begin{definition}\label{def:M-general}
    The set of all derivations to $S[i,E(x_j,i,k_T,k_S)]$ rooted in $x_j$ that have exactly $k_T$ deletions from the tree and exactly $k_S$ deletions from the string is denoted by $\Dm{}$.
\end{definition}

\begin{lemma}\label{lemma:main-proof}
    At the end of the algorithm every entry $\A[j,i,k_T,k_S]$ of the DP-table $\A{}$ holds the highest score of a derivation of $S[i,E(x_j,i,k_T,k_S)]$ rooted in $x_j$ that has $k_S$ deletions from the string and $k_T$ deletions from the tree, i.e. $\A[j,i,k_T,k_S] = \max_{\mu \in \Dm{}}\mu.score$
\end{lemma}

\begin{proof}
We prove \cref{lemma:main-proof} by induction on the entries of $\A{}$ in the order described in the algorithm. 
Namely, for two entries $\A[j_1,i_1, k_{T_1},k_{S_1}]$ and $\A[j_2,i_2,k_{T_2},k_{S_2}]$, $\A[j_1,i_1, k_{T_1},k_{S_1}] < \A[j_2,i_2,k_{T_2},k_{S_2}]$ \ifff{} $j_1<j_2$ or both $j_1=j_2$ and $i_1<i_2$. 
If $j_1=j_2$ and $i_1=i_2$, then the order between the entries is chosen arbitrarily.

\subparagraph{Base Case.}
The base case of the algorithm is the initialization of the DP table, where the entries $\A[j,i,k_T,k_S]$ for $x_j\in \leaves{root}$ and $k_T \in \{0,1\}$ are computed. 
When $k_T=0$, there are no deletions from the tree. So, $x_j$ must be mapped to some character $S[\ell]$ ($i\leq \ell \leq E(x_j,i,0,k_S)$). In this version of the algorithm the deletion of a character does not change the score of the derivation, so the maximal score of a derivation in $\D_M(x_{j},i,0,k_S)$ is the maximum score of a mapping of $x_j$ to some character $S[\ell]$ ($i\leq \ell \leq E(x_j,i,0,k_S)$), which is the initialization value of the entry $\A[j,i,0,k_S]$.
When $k_T=1$, there is one deletion from the tree. The derived subtree $T(x_j)$ has one leaf, $x_j$, and so it must be the deleted leaf. All characters in the derived string, $S[i:E(x_j,i,1,k_S)]$, must also be deleted. Deletions do not add to the score of the derivation, and so all the derivations in $\D_M(x_{j},i,1,k_S)$ have a score of $0$, which is the initialization value of $\A[j,i,1,k_S]$.

\subparagraph{Induction Assumption.}
Assume that every entry $\A[j',i',k'_T,k'_S]$ such that $\A[j',i',k'_T,k'_S]$ $< \A[j,i,k_T,k_S]$ holds the best score of a derivation from $\D_M(x_{j'},i',k'_T,k'_S)$. Namely, $\A[j',i',k'_T,k'_S] = \max_{\mu\in \D_M(x_{j'},i',k'_T,k'_S)}{\mu.score} = OPT(j',i',k'_T,k'_S)$.

\subparagraph{Induction Step.}
For every internal node $x_j$ and possible start index $i$, the algorithm fills the DP table entry $\A[j,i,k_T,k_S]$ according to the values returned from the Q-mapping and P-mapping algorithms according to the type of $x_j$. The correctness of these algorithms is proven in \cref{sec:q-node-proof} and \cref{sec:p-node-proof}, respectively. Hence, it is only necessary to prove that the input the algorithms expect to receive is sent correctly from the main algorithm.

Both the Q-mapping and P-mapping algorithms expect to receive the internal node which should be the root of all the output derivations, a substring $S'$ of $S$, the deletion limits $d_T$ and $d_S$, and a collection of the best scoring derivations of every child of $x$ to every substring of $S'$ with up to $d_T$ and $d_S$ deletions from the tree and string, respectively.
By definition an entry in $\A[j,i,\cdot,\cdot]$ concerns the derivations of $x_j$ with a start point $i$. The end point of the longest derivation of those derivations is $E(j,i,0,d_S)$. Hence, the internal node sent to the Q-mapping or P-mapping algorithm is $x_j$ and the substring $S'$ equals $S[i,E(j,i,0,d_S)]$. The deletion limits $d_T$ and $d_S$ are given as input to the main algorithm.
Lastly, the best derivations of the children of $x_j$ are stored in $\A{}$. Because the nodes of $T$ are indexed in postorder, if $x_c$ is a child of $x_j$, then $c < j$. Hence, for every $i', k'_T, k'_S$, it holds that $\A[c,i',k'_T,k'_S] < \A[j,i,k_T,k_S]$, and from the induction assumption $\A[c,i',k'_T,k'_S]= OPT(c,i',k'_T,k'_S)$. So, indeed the expected input to the Q-mapping and P-mapping algorithms is correct.
This completes the proof.
\end{proof}

%% file: algorithm/proof/p-proof.tex
\subsection{P-Node Mapping}\label{sec:p-node-proof}
In this section we give the pseudocode of the P-mapping algorithm presented in \cref{subsec:p-mapping-alg} (\cref{alg:p-node}) and prove its correctness.

\begin{algorithm}[t]
\SetAlgoLined
\SetKwComment{Comment}{//}{}
\KwIn{$x, S', \D{}, d_T, d_S$}
\KwOut{The best derivation of $x$ to every prefix of $S'$}
$\gamma \gets |\children{x}|$\;
Build $\PP{}$ with dimensions $2^\gamma \times d_T+1 \times d_S+1$\;
\For{$size=0$ \KwTo $\gamma$}{
    \ForEach{$C\subseteq \children{x}$ s.t. $|C|=size$}{
        \For{$k_S=0$ \KwTo $d_S$}{
    	    \For{$k_T=0$ \KwTo $d_T$}{
        	    \eIf{($L(C,k_T,k_S) = 0$ and $k_S=0$) or ($size = 0$ and $k_T=0$)} 
        	    {\Comment{initialization}
        	    $\PP[C,k_T,k_S] \gets 0$\;}
            	{compute $\PP[C,k_T,k_S]$ according to \cref{eq:p-recursion}\;}
	        }
    	}
    }
 }
 \KwRet{$\PP[\children{x},\cdot,\cdot]$} \;
\caption{P-Mapping.}
\label{alg:p-node}
\end{algorithm}

\begin{lemma}\label{lemma:p-node}
    At the end of the algorithm every entry of the DP-table, $\PP[C,k_T,k_S]$, holds the best score for a derivation of $\xC{}$ to a prefix of $S'$ with $k_T$ deletions from the tree and $k_S$ deletions from the string, i.e. $\PP[C,k_T,k_S] = \max_{\mu \in \Mxc{}}\mu.score$
\end{lemma}

\begin{proof}
We prove \cref{lemma:p-node} by induction on the entries of $\PP{}$ in the order described in the algorithm. Namely, for two entries $\PP[C_1,k_{T_1},k_{S_1}]$ and $\PP[C_2,k_{T_2},k_{S_2}]$, $\PP[C_1,k_{T_1},k_{S_1}] < \PP[C_2,k_{T_2},k_{S_2}]$ \ifff 
\begin{itemize}
    \item $|C_1| < |C_2|$, or
    \item $|C_1| = |C_2|$ and $k_{S_1} < k_{S_2}$, or 
    \item $|C_1| = |C_2|$ and $k_{S_1} = k_{S_2}$ and $k_{T_1} < k_{T_2}$ 
\end{itemize}
If $C_1 \neq C_2$, $|C_1| = |C_2|$, $k_{S_1} = k_{S_2}$ and $k_{T_1} = k_{T_2}$ are all satisfied, then the order between the entries is chosen randomly.

\subparagraph{Base Cases.} 
There are two types of base cases, as described in the initialization of the DP table. 
\begin{enumerate}
    \item $L(C,k_T,k_S)=0$ and $k_S=0$: Let $\mu$ be a derivation of $\xC{}$ with $k_T$ and $k_S$ deletions. By definition, $\mu$ derives an empty string, i.e. there are no characters to map to the leaves of the subtrees rooted in the nodes in $C$. Hence, every child of $x$ that is considered (the nodes in $C$) must be deleted under $\mu$. All the nodes in $C$ can be deleted if the sum of their spans is equal to the allowed number of deletions in $\mu$ (that is, $k_T$). From the definition of $L(C,k_T,k_S)=0$ and the fact that $k_S=0$, we receive that indeed $k_T = \sum_{c\in C}\xspan{c}$. Every child node of $x$ that is kept under $\mu$ adds to the score of the derivation of $x$, but there are none in this case. In addition, every deletion from the subtree $T(x)$ adds nothing to the score (in the penalization-free version of the algorithm). Hence, the score of $\mu$ must equal $0$.
    \item $C=\emptyset$ and $k_T=0$: In this case all of the children of $x$ are ignored, so there are no leaves to map. Hence, every character of the derived string should be deleted. Note that, the derived string is $S'[1:\E(C,k_T,k_S)]$, and its length is $L(C,k_T,k_S) = \sum_{c\in C}\xspan{c} - k_T + k_S = \sum_{c\in \emptyset}\xspan{c} - 0 + k_S = k_S$. So, the number of deletions from the string in this state is exactly the number needed to delete the derived string.
\end{enumerate}

\subparagraph{Induction Assumption.}
Assume that every table entry $\PP[C',k'_T,k'_S]$ such that $\PP[C',k'_T,k'_S]$ $< \PP[C,k_T,k_S]$ holds the best score of a derivation in $\D(x^{(C')},k'_T,k'_S)$. Namely, $\PP[C',k'_T,k'_S]$ $= \max_{\mu\in \D(\xC{'},k'_T,k'_S)}{\mu.score} = OPT(C',k'_T,k'_S)$.

\subparagraph{Induction Step.}
Towards the proof of the step, we prove the following \cref{eq:p-induction-step}:
\begin{equation}\label{eq:p-induction-step}
\begin{split}
   OPT(C,k_T,k_S) = &\max(OPT(C,k_T,k_S-1),\\
   &\displaystyle\max_{\mu \in \Mc{}} {OPT(C\setminus \{\mu.v\},k_T-\mu.\dt,k_S-\mu.\ds) + \mu.score}) 
\end{split}
\end{equation}
\begin{itemize}
    \item[$\leq$:]
    Let $\mu^* \in \Mxc$ be a derivation such that $\mu^*.score = OPT(C,k_T,k_S)$.
    By definition, $\mu^*$ is a derivation of $\xC{}$ to the string $S'[1:\E(C,k_T,k_S)]$. In a derivation every character of the derived string is either deleted or it is a part of a substring derived from one of the children of $x$. So, either $S'[\E(C,k_T,k_S)]$ is deleted under $\mu^*$, or it is mapped under some derivation of a child of $x$, $y \in C$, to a substring $S'[i:\E(C,k_T,k_S)]$ (for an index $0<i\leq\E(C,k_T,k_S)$).
    
    First, if the former is true, then by removing the deletion of $S'[\E(C,k_T,k_S)]$ from $\mu^*$, $\rmdel{\mu^*}{\E(C, k_T, k_S)}$, a derivation $\mu' \in \D(x^{(C)},k_T,k_S-1)$ is received. The derivation $\mu'$ derives the string $S'[1:\E(C,k_T,k_S-1)] = S'[1:\E(C,k_T,k_S)] - 1$. So, the following \cref{eq:leq-string-deletion} is true.
    \begin{equation} \label{eq:leq-string-deletion}
    \begin{split}
        \mu^*.score &= \mu'.score\\
        &\leq OPT(C,k_T,k_S-1)\\
        &\leq \max (OPT(C,k_T,k_S-1),\\ &\displaystyle\max_{\mu \in \Mc{}}
        {OPT(C\setminus \{\mu.v\},k_T-\mu.\dt,k_S-\mu.\ds) + \mu.score})
    \end{split}
    \end{equation}
    Note that even if there is a penalization cost for deletions, the cost for the deletion of $S'[\E(C,k_T,k_S)]$ (i.e. $-\Delta(S'[\E(C,k_T,k_S)])$) is constant in this setting. So, for two derivations $\eta, \eta' \in \D(k_T,k_S-1,x^{(C)})$ if $\eta.score \leq \eta'.score$ then $\eta.score -\Delta(S'[\E(C,k_T,k_S)]) \leq \eta'.score -\Delta(S'[\E(C,k_T,k_S)])$. Hence, the conclusion from \cref{eq:leq-string-deletion} is still true.

    Second, if the latter is true, then there is a node $y \in C$ for which there is a derivation $\mu_y \in \D{}$ such that $\mu_y.e = \E(C,k_T,k_S)$ and $\mu.y$ is the derivation of $y$ under $\mu^*$.
    For $\mu^*$ to be a legal derivation, $\mu_y$ must be in $\Mc{}$. Hence, $\mu_y.score \leq \max_{\mu \in \Mc{}}{\mu.score}$.
    
    Furthermore, by removing $\mu_y$ from $\mu^*$, $\remove{\mu^*}{\mu.y}$, the received partial derivation, $\mu'$, is of $\xC{\setminus\{y\}}$ to $S'[1:\mu_y.s-1]$ with $k_T-\mu_y.\dt$ deletions from the tree and $k_S-\mu_y.\ds$ from the string. Thus, $\mu' \in \D(k_T-\mu_y.\dt, k_S-\mu_y.\ds, \xC{\setminus\{y\}})$, and so $\mu'.score \leq OPT(\xC{\setminus\{y\}}, k_T-\mu_y.\dt, k_S-\mu_y.\ds)$.
    Note that, indeed $\mu_y.s = 1 + \E(C\setminus\{y\},k_T-\mu_y.\dt, k_S-\mu_y.\ds)$, as can be seen in the following \cref{eq:proof-start-end-points-match}.
    \begin{equation} \label{eq:proof-start-end-points-match}
    \begin{split}
        \mu_y.s &= \E(C,k_T,k_S) - L(y,\mu_y.\dt, \mu_y.\ds) + 1 \\
        &= \sum_{c\in C}{\xspan{c}} + k_S - k_T - (\mu_y.\ds - \mu_y.\dt + \xspan{y}) + 1 \\ 
        &= \sum_{c\in C\setminus\{y\}}{\xspan{c}} + k_S-\mu_y.\ds - (k_T-\mu_y.\dt) + 1 \\
        &= \E(C\setminus\{y\},k_T-\mu_y.\dt, k_S-\mu_y.\ds) + 1
    \end{split}
    \end{equation}
    By combining our conclusions about $\mu_y$ and $\mu'$ together, we receive the following \cref{eq:mapping-leq-opt}. 
    \begin{equation} \label{eq:mapping-leq-opt}
        \begin{split}
            \mu^*.score &= \mu'.score + \mu_y.score \\
            &\leq OPT(C\setminus\{y\},k_T-\mu_y.\dt, k_S-\mu_y.\ds) + \max_{\mu \in \Mc{}}{\mu.score} \\
            &\leq \displaystyle\max_{\mu \in \Mc{}}
            {OPT(C\setminus \{\mu.v\},k_T-\mu.\dt,k_S-\mu.\ds) + \mu.score} \\
            &\leq \max (OPT(C, k_T, k_S-1), \\ 
            &\displaystyle\max_{\mu \in \Mc{}}
            {OPT(C\setminus \{\mu.v\},k_T-\mu.\dt,k_S-\mu.\ds) + \mu.score})
        \end{split}
    \end{equation}
    
    \item[$\geq$:]
    Let $\mu^*$ be a derivation such that \cref{eq:p-proof-geq} holds.
    \begin{equation} \label{eq:p-proof-geq}
    \begin{split}
        \mu^*.score = \max(&OPT(C, k_T, k_S-1), \\
        &\max_{\mu \in \Mc{}} OPT(C\setminus \{\mu.v\},k_T-\mu.\dt,k_S-\mu.\ds) + \mu.score)
    \end{split}
    \end{equation}
    So, either $\mu^*.score = OPT(C, k_T, k_S-1)$, 
    or $\mu^*.score = \displaystyle\max_{\mu \in \Mc{}}$
            $OPT(C\setminus\{\mu.v\},$ $k_T-\mu.\dt,k_S-\mu.\ds) + \mu.score$.
    
    First, if the former is true, let $\eta \in \D(\xC{},k_T,k_S-1)$ be a derivation with $\eta.score = OPT(C, k_T, k_S-1)$. By definition, $\eta$ derives the substring $S'[1:\E(C,k_T,k_S-1)]$. Adding to $\eta$ the deletion of $S'[\E(C,k_T,k_S)]$, $\adddel{\eta}{\E(C,k_T,k_S)}$, results in a derivation $\eta'$ of $\xC{}$ to the string $S'[1:\E(C,k_T,k_S)]$ with $k_T$ deletions from the tree and $k_S$ deletions from the string. The string $S'[1:\E(C,k_T,k_S)]$ is equal to the concatenation of $S'[1:\E(C,k_T,k_S-1)]$ and $S'[\E(C,k_T,k_S)]$. So, $\eta' \in \Mxc{}$, and thus $\eta'.score \leq OPT(C, k_T, k_S)$. The derivation $\eta'$ was constructed such that $\mu^*.score = \eta'.score$, so $\mu^*.score \leq OPT(C, k_T, k_S)$.
    
    Second, if the latter is true, then let $\eta^* = \argmax_{\mu \in \Mc{}}OPT(C\setminus \{\mu.v\}, k_T - \mu.\dt, k_S - \mu.\ds) + \mu.score$. Adding $\eta^*$ to a partial derivation $\eta \in \D(\xC{\setminus \{\eta^*.v\}},k_T - \eta^*.\dt, k_S - \eta^*.\ds)$, $\add{\eta}{\eta^*}$, results in a partial derivation, $\eta'$, with $k_T-\eta^*.\dt + \eta^*.\dt = k_T$ deletions from the tree and $k_S-\eta^*.\ds + \eta^*.\ds = k_S$ deletions from the string, that takes into account the children of $x$ that are in $C\setminus \{\eta^*.v\} \cup \{\eta^*.v\} = C$. It is a legal partial derivation since $\eta^*$ derives the node $\eta^*.v$ that is not in $C\setminus \{\eta^*.v\}$ to a string that does not intersect with the string derived by $\eta$. 
    The string that is derived by $\eta$ is $S'[\eta.s:\eta.e]$ and it does not intersect with the string derived by $\eta^*$ ($S'[\eta^*.s:\eta^*.e]$). That is because $\eta.e + 1 = \eta^*.s$, as can be seen similarly to \cref{eq:proof-start-end-points-match}. So, $\eta' \in \Mxc$, and thus $\eta'.score \leq OPT(C, k_T, k_S)$. The partial derivation $\eta'$ was constructed such that $\mu^*.score = \eta'.score$, so $\mu^*.score \leq OPT(C, k_T, k_S)$.
\end{itemize}
 From the induction assumption, $\PP[C, k_T, k_S-1] = OPT(C, k_T, k_S-1)$ and for every $\mu \in \Mc$, $\PP[C\setminus \{\mu.v\}, k_T - \mu.\dt, k_S - \mu.\ds] = OPT(C\setminus \{\mu.v\}, k_T - \mu.\dt, k_S - \mu.\ds)$. Thus from \cref{eq:p-induction-step}, it follows that $\PP[C, k_T, k_S] = OPT(C, k_T, k_S)$. This completes the proof.
\end{proof}

%% file: algorithm/proof/q-proof.tex
\subsection{Q-Node Mapping}\label{sec:q-node-proof}
In this section we prove the correctness of the Q-mapping algorithm presented in \cref{sec:q-node-mapping}. We prove the algorithm for the case where the children of $x$ can only be arranged in a left-to-right order. The proof of the right-to-left order is similar.
\begin{lemma}\label{lemma:q-node}
    At the end of the algorithm every entry of the DP-table $\Q{}$, $\Q[i,k_T,k_S]$, holds the best score of a derivation of $\x{i}$ and a prefix of $S'$ with $k_T$ deletions from the tree and $k_S$ deletions from the string, i.e. $\Q[i,k_T,k_S] = \max_{\mu \in \Mxi{}}\mu.score$.
\end{lemma}

\begin{proof}
We prove \cref{lemma:q-node} by induction on the entries of $\Q{}$ in the order described in the algorithm. Namely, for two entries $\Q[i_1, k_{T_1},k_{S_1}]$ and $\Q[i_2,k_{T_2},k_{S_2}]$, $\Q[i_1,k_{T_1},k_{S_1}] < \Q[i_2,k_{T_2},k_{S_2}]$ \ifff
\begin{itemize}
    \item $i_1 < i_2$, or
    \item $i_1 = i_2$ and $k_{S_1} < k_{S_2}$, or
    \item $i_1 = i_2$ and $k_{S_1} = k_{S_2}$ and $k_{T_1} < k_{T_2}$.
\end{itemize}

\subparagraph{Base Case.}
The base case is the initialization of the DP table entries $\Q[0,0,k_S]$ for $0\leq k_S \leq \len{S'}$, with a value of $0$. Each of these entries holds the score of some derivation $\mu$ of $\x{0}$, i.e. $\mu$ is a partial derivation that ignores all nodes in $T(x)$. In addition $\mu$ derives the substring $S'[1:L(\x{0},0,k_S)] = S'[1:k_S]$. Hence, all the characters in $S'[1:k_S]$ must be deleted under $\mu$. Each deletion does not add to the score of the derivation and there are no mappings under $\mu$ either, so the score of such a derivation is $0$.

\subparagraph{Induction Assumption.}
Assume that every table entry $\Q[i',k'_T,k'_S]$ such that $\Q[i',k'_T,k'_S] < \Q[i,k_T,k_S]$ holds the best score of a derivation from $\D(\x{i'},k'_T,k'_S)$. Namely, $\Q[i',k'_T,k'_S]= \max_{\mu\in \D(\x{i'},k'_T,k'_S)}{\mu.score} = OPT(i',k'_T,k'_S)$.

\subparagraph{Induction Step.}
Towards the proof of the step, we prove the following \cref{eq:q-induction-step}:
\begin{equation}\label{eq:q-induction-step}
\begin{split}
   OPT(i,k_T,k_S) = \max(&OPT(i,k_T, k_S-1),\\ &OPT(i-1,k_T-\xspan{x_i}, k_S),\\
   &\displaystyle\max_{\substack {\mu \in \Mi{} \\ \mathrm{s.t.}\ \mu .v = x_i}} OPT(i-1, k_T-\mu .\dt,k_S-\mu .\ds) + \mu .score) 
\end{split}
\end{equation}

\begin{itemize}
    \item [$\leq$:]
    Let $\mu^* \in \Mxi$ be a derivation such that $\mu^*.score = OPT(i,k_T,k_S)$.
    By definition, $\mu^*$ is a derivation of $\x{i}$ to the string $S'[1:\E(i,k_T,k_S)]$. Under a derivation every child of the root of the derivation is either deleted or kept and every character of the derived string is either deleted or mapped. Thus, $x_i \in \children{\x{i}}$ is either deleted or kept under $\mu^*$, and the character $S'[\E(i,k_T,k_S)]$ is either deleted or mapped under $\mu^*$. Now, let us consider every case.
    
    First, consider a case in which $x_i$ is deleted under $\mu^*$. By removing the deletion of $x_i$ from $\mu^*$ ($\rmdel{\mu^*}{x_i}$) a partial derivation, $\mu'$, that ignores $x_i$ is received, therefore the root of $\mu'$ is $\x{i-1}$. By \cref{def:rm-node-del}, $\mu'$ has $\mu^*.\dt-\xspan{x_i} = k_T - \xspan{x_i}$ deletions from the tree and $\mu^*.\ds=k_S$ deletions from the string. Hence, $\mu' \in \D(\xis{-1},k_T - \xspan{x_i},k_S)$ and \cref{eq:leq-node-deletion} below is true (remember that a deletion of a node does not change the score of a derivation).
    \begin{equation} \label{eq:leq-node-deletion}
    \begin{split}
        \mu^*.score = \mu'.score 
        &\leq OPT(i-1,k_T-\xspan{x_i}, k_S)\\
        &\leq \max (OPT(i,k_T, k_S-1), OPT(i-1,k_T-\xspan{x_i}, k_S), \\
        &\displaystyle\max_{\substack {\mu \in \Mi{} \\ \mathrm{s.t.}\ \mu .v = x_i}}
        {OPT(i-1,k_T-\mu.\dt,k_S-\mu.\ds) + \mu.score})
    \end{split}
    \end{equation}
    
    Second, consider a case in which $S'[\E(i,k_T,k_S)]$ is deleted under $\mu^*$. By removing the deletion of $S'[\E(i,k_T,k_S)]$ from $\mu^*$, $\rmdel{\mu^*}{\E(i,k_T,k_S)}$ (see \cref{def:rm-char-del}), the partial derivation received, $\mu'$, has $k_T$ and $k_S-1$ deletions from the tree and string, respectively, and its root is $\x{i}$. Hence, $\mu' \in \D(\x{i},k_T,k_S-1)$ and \cref{eq:q-leq-string-deletion} below is true (remember that a deletion of a character does not change the score of a derivation).
    \begin{equation} \label{eq:q-leq-string-deletion}
    \begin{split}
        \mu^*.score = \mu'.score 
        &\leq OPT(i,k_T, k_S-1)\\
        &\leq \max (OPT(i,k_T, k_S-1), OPT(i-1,k_T-\xspan{x_i}, k_S), \\
        &\displaystyle\max_{\substack {\mu \in \Mi{} \\ \mathrm{s.t.}\ \mu .v = x_i}}
        {OPT(i-1,k_T-\mu.\dt,k_S-\mu.\ds) + \mu.score})
    \end{split}
    \end{equation}
    
    Lastly, if neither is true, then $S'[\E(i,k_T,k_S)]$ is mapped under $\mu^*$ and $x_i$ is kept under $\mu^*$. Let $\mu_i$ be the derivation of $x_i$ under $\mu^*$ (there is one because $x_i$ is kept under $\mu^*$). Because $S'[\E(i,k_T,k_S)]$ is mapped, then it is a part of a substring of $S'[1:\E(i,k_T,k_S)]$ that is derived by some derivation, $\mu_j$, such that $\mu_j$ is the derivation of the child node $\mu_j.v \in \children{\x{i}}$ under $\mu^*$. 
    Since $x_i$ is the rightmost child of $\x{i}$ and the children of the Q-node $x$ can only be arranged from left to right (in this proof), $\mu_j.v$ must be $x_i$. Otherwise, the left-to-right ordering is defied. Every child of $x$ can only have one derivation under $\mu^*$, so $\mu_i=\mu_j$.
    Note that $\mu_i$ must have up to $k_T$ and $k_S$ deletions from the tree and string, respectively, else $\mu^*$ is not a legal derivation. In addition, the end point of $\mu_i$ is $\E(i,k_T,k_S)$. Let $\mu_i^*$ be the highest scoring derivation of $x_i$ with up to $k_T$ and $k_S$ deletions which has the endpoint $\E(i,k_T,k_S)$, i.e $\mu_i.score \leq \mu_i^*.score$. By definition, $\mu_i^* \in \D_\leq(\x{i},k_T,k_S)$, hence $\mu_i.score \leq \max_{\mu \in \D_\leq(\x{i},k_T,k_S)}\mu.score$.
    Now, removing $\mu_i$ from $\mu^*$, $\remove{\mu^*}{\mu_i}$ (see \cref{def:remove-derivation}), results in a derivation, $\mu'$, with $\mu^*.\dt-\mu_i.\dt = k_T - mu_i.\dt$ deletions from the tree and $\mu^*.\ds-\mu_i.\ds = k_S - mu_i.\ds$ deletions from the string. In addition $\mu'$ ignores $x_i$, and so its root is $\x{i-1}$. Hence, similarly to $\mu_i$, $\mu'.score \leq \max_{\mu \in\D(i-1,k_T - mu_i.\dt,k_S - mu_i.\ds)}\mu.score = OPT(i-1,k_T- mu_i.\dt,k_S- mu_i.\ds)$. 
    Putting the conclusions on $\mu_i$ and $\mu'$ together we receive \cref{eq:q-leq-max-derivation} below.
    \begin{equation}\label{eq:q-leq-max-derivation}
    \begin{split}
        \mu^*.score &= \mu_i.score + \mu'.score \\
        &\leq \max_{\mu \in \D_\leq(\x{i},k_T,k_S)}\mu.score + OPT(i-1,k_T- \mu_i.\dt,k_S- \mu_i.\ds) \\
        &\leq \max_{\substack {\mu \in \Mi{} \\ \mathrm{s.t.}\ \mu .v = x_i}}
        {OPT(i-1,k_T-\mu.\dt,k_S-\mu.\ds) + \mu.score}) \\
        &\leq \max (OPT(i,k_T, k_S-1), OPT(i-1,k_T-\xspan{x_i}, k_S), \\
        &\displaystyle\max_{\substack {\mu \in \Mi{} \\ \mathrm{s.t.}\ \mu .v = x_i}}
        {OPT(i-1,k_T-\mu.\dt,k_S-\mu.\ds) + \mu.score})
    \end{split}
    \end{equation}
    
    In any case \cref{eq:q-leq} below is true.
    \begin{equation} \label{eq:q-leq}
        \begin{split}
            OPT(i,k_T,k_S) &= \mu^*.score \\
            &\leq \max (OPT(i,k_T, k_S-1), OPT(i-1,k_T-\xspan{x_i}, k_S), \\ &\displaystyle\max_{\substack {\mu \in \Mi{} \\ \mathrm{s.t.}\ \mu .v = x_i}}
        {OPT(i-1,k_T-\mu.\dt,k_S-\mu.\ds) + \mu.score})
        \end{split}
    \end{equation}
    
    \item [$\geq$:]
    Let $\mu^*$ be a derivation such that $\mu^*.score = \max(OPT(i,k_T, k_S-1), OPT(i-1,k_T-\xspan{x_i}, k_S),
   \displaystyle\max_{\substack {\mu \in \Mi{} \\ \mathrm{s.t.}\ \mu .v = x_i}} OPT(i-1, k_T-\mu .\dt,k_S-\mu .\ds) + \mu .score)$. Hence, $\mu^*.score = OPT(i,k_T, k_S-1)$ or $\mu^*.score = OPT(i-1,k_T-\xspan{x_i}, k_S)$ or $\mu^*.score = \displaystyle\max_{\substack {\mu \in \Mi{} \\ \mathrm{s.t.}\ \mu .v = x_i}} OPT(i-1, k_T-\mu .\dt,k_S-\mu .\ds) + \mu .score$. 
   
   First, assume $\mu^*.score = OPT(i,k_T, k_S-1)$.
   Let $\eta \in \D(\x{i},k_T,k_S-1)$ be a derivation with $\eta.score = OPT(i, k_T, k_S-1)$. By definition, $\eta$ derives the substring $S'[1:\E(\xis{},k_T,k_S-1)]$. From \cref{def:add-char-del}, Adding to $\eta$ the deletion of $S'[\E(\xis{},k_T,k_S)]$ ($\adddel{\eta}{\E(\xis{},k_T,k_S)}$) results in a derivation, $\eta'$ that derives $\x{i}$ to the string $S'[1:\E(\xis{},k_T,k_S)]$ with $k_T$ deletions from the tree and $k_S$ deletions from the string. The string $S'[1:\E(\xis{},k_T,k_S)]$ is equal to the concatenation of $S'[1:\E(\xis{},k_T,k_S-1)]$ and $S'[\E(\xis{},k_T,k_S)]$. So, $\eta' \in \Mxi{}$, and thus $\eta'.score \leq OPT(i,k_T, k_S)$. We have thus built $\eta'$ such that $\mu^*.score = \eta'.score$, so $\mu^*.score \leq OPT(k_T, k_S, C)$.
   
   Second, assume $\mu^*.score = OPT(i-1,k_T-\xspan{x_i}, k_S)$. Let $\eta \in \D(\x{i-1},k_T-\xspan{x_i},k_S)$ be a derivation with $\eta.score = OPT(i-1, k_T-\xspan{x_i}, k_S)$. By definition, $\eta$ derives the substring $S'[1:\E(\xis{-1},k_T-\xspan{x_i},k_S)]$. From \cref{def:rm-node-del}, adding the deletion of the node $x_i$ to $\eta$ ($\adddel{\eta}{x_i}$) results in a derivation $\eta'$ that derives $\x{i}$ to the string $S'[1:\E(\xis{},k_T,k_S)]$ with $k_T$ deletions from the tree and $k_S$ deletions from the string. So, $\eta' \in \Mxi{}$, and thus $\eta'.score \leq OPT(i,k_T, k_S)$. We built $\eta'$ such that $\mu^*.score = \eta'.score$, so $\mu^*.score \leq OPT(k_T, k_S, C)$.
   
   Lastly, assume $\mu^*.score = \displaystyle\max_{\substack {\mu \in \Mi{} \\ \mathrm{s.t.}\ \mu .v = x_i}} OPT(i-1, k_T-\mu .\dt,k_S-\mu .\ds) + \mu .score$. 
   Let $\eta^*$ be a derivation of $x_i$ that is in $\Mi{}$ that by setting $\eta^*$ to $\mu$ yields the highest value for $OPT(i-1, k_T-\mu .\dt,k_S-\mu .\ds) + \mu .score$ of all the derivations of $x_i$ in $\Mi{}$. Formally, $\eta^* = \argmax_{\substack {\mu \in \Mi{} \\ \mathrm{s.t.}\ \mu .v = x_i}} OPT(i-1, k_T-\mu .\dt,k_S-\mu .\ds) + \mu .score$. 
   From \cref{def:add-derivation}, adding $\eta^*$ to a partial derivation $\eta \in \D(\xis{-1},k_T - \eta^*.\dt, k_S - \eta^*.\ds)$ ($\add{\eta}{\eta^*}$) results in a partial derivation, $\eta'$, with $k_T-\eta^*.\dt + \eta^*.\dt = k_T$ deletions from the tree and $k_S-\eta^*.\ds + \eta^*.\ds = k_S$ deletions from the string, that takes into account the children of $x$ that are in $\xis{-1} \cup \{x_i\} = \xis{}$. It is a legal partial derivation since $\eta^*$ derives the node $\eta^*.v$ that is not in $\xis{-1}$ to a string that does not intersect with the string derived by $\eta$. 
   The string that is derived by $\eta$ is $S'[\eta.s:\eta.e]$ and it does not intersect with the string derived by $\eta^*$ ($S'[\eta^*.s:\eta^*.e]$). That is because $\eta.e + 1 = \eta^*.s$, as can be seen in \cref{eq:q-proof-start-end-points-match} below. So, $\eta' \in \Mxi{}$, and thus $\eta'.score \leq OPT(i,k_T, k_S)$. We built $\eta'$ such that $\mu^*.score = \eta'.score$, so $\mu^*.score \leq OPT(i,k_T, k_S)$.
    \begin{equation}\label{eq:q-proof-start-end-points-match}
        \begin{split}
            \eta^*.s &= \E(\xis{},k_T,k_S) - L(x_i,\eta^*.\dt, \eta^*.\ds) + 1 \\
                &= \sum_{x_j\in \xis{}}{\xspan{x_j}} + k_S - k_T - (\eta^*.\ds - \eta^*.\dt + \xspan{x_i}) + 1 \\ 
                &= \sum_{x_j\in \xis{-1}}{\xspan{x_j}} + k_S-\eta^*.\ds - (k_T-\eta^*.\dt) + 1 \\
                &= \E(\xis{-1},k_T-\eta^*.\dt, k_S-\eta^*.\ds) + 1 = \eta.e + 1
        \end{split}
    \end{equation}
\end{itemize}
 From the induction assumption, $\Q[i,k_T,k_S-1] = OPT(i,k_T,k_S-1)$, $\Q[i-1,k_T-\xspan{x_i},k_S] = OPT(i-1,k_T-\xspan{x_i},k_S)$ and for every $\mu \in \Mxi{}$ such that $\mu.v=x_i$, $\Q[i-1,k_T - \mu.\dt, k_S - \mu.\ds] = OPT(i-1,k_T - \mu.\dt, k_S - \mu.\ds)$. Thus from \cref{eq:q-induction-step}, it follows that $\Q[i,k_T,k_S] = OPT(i,k_T,k_S)$. This completes the proof.
\end{proof}

%% file: appendix-figures.tex
\section{Figures}\label{sec:appendix-figs}

\begin{figure}[ht]
    \centering
    \begin{subfigure}[b]{0.3\textwidth}
        \centering
        \includegraphics[width=0.8\textwidth]{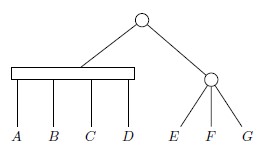}
        \caption{$T_1$}
        \label{fig:pqt1}
    \end{subfigure}
    \hfill
    \begin{subfigure}[b]{0.3\textwidth}
        \centering
        \includegraphics[width=0.8\textwidth]{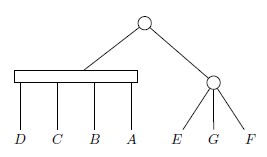}
        \caption{$T_2$}
        \label{fig:pqt2}
    \end{subfigure}
    \hfill
    \begin{subfigure}[b]{0.3\textwidth}
        \centering
        \includegraphics[width=0.8\textwidth]{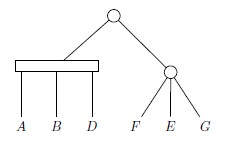}
        \caption{$T_3$}
        \label{fig:pqt3}
    \end{subfigure}
    \caption{Three different \pqt{s}. $T_2$ can be obtained from $T_1$ by reversing the children of a Q-node (the left child of the root) and by reordering the children of a P-node (the right child of the root), so $T_2 \equiv T_1$. $T_3$ can be obtained from $T_1$ by deleting one leaf and permuting the children of the right child of the root, so $T_1 \succeq_1 T_3$. Now, $T_2 \succeq_1 T_3$ can be inferred, because the $\equiv$ is an equivalence relation.
    By the definition of frontier, $F(T_1)=ABCDEFG$; $F(T_2)=DCBAEGF$; $F(T_3)=ABDFEG$.}
    \label{equiv-pqts}
\end{figure}

\begin{figure}[ht!]
    \centering
    \begin{subfigure}[b]{1\textwidth}
        \includegraphics[width=\textwidth]{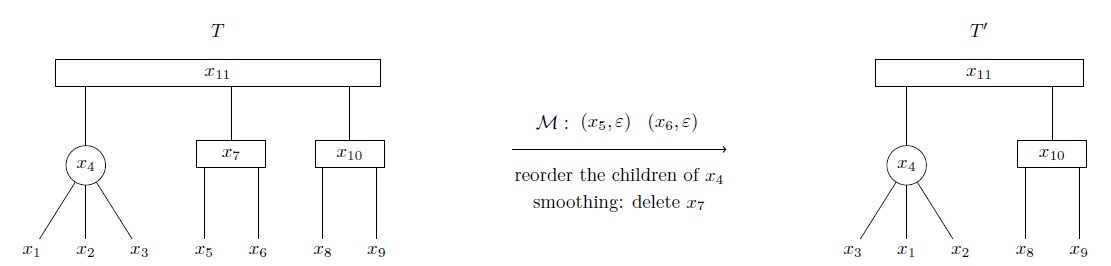}
        \caption{The derivation $\mu$ applied on $T$ resulting in $T'$: reorder the children of $x_4$, delete leaves according to $\M{}$ (delete $x_5$ and $x_6$) and preform smoothing (delete $x_7$, the parent node of $x_5$ and $x_6$). The root of $T$, $x_{11}$, is the node that $\mu$ derives, denoted $\mu.v$. Also, $\mu$ is a derivation of $x_{11}$. The nodes $x_5$, $x_6$ and $x_7$ are deleted under $\mu$. The leaves $x_1,x_2,x_3,x_8,x_9$ are mapped under $\mu$. The nodes $x_4, x_{10}, x_{11}$ are kept under $\mu$.}
        \label{fig:derivation-tree}
    \end{subfigure}
    
    \begin{subfigure}[b]{1\textwidth}
        \centering
        \includestandalone[width=\textwidth]{figures/derivation-string}
        \caption{The derivation $\mu$ on $S'$ resulting in $S_\M{}$: apply substitutions and deletions according to $\M{}$. The substring $S'=S[3:8]$ is the string that $\mu$ derives. The character $S[4]$ is deleted under $\mu$. The characters $S[3],S[5],S[6],S[7],S[8]$ are mapped under $\mu$.}
        \label{fig:derivation-string}
    \end{subfigure}
    \caption{An illustration of the derivation $\mu$ from the \pqt{} $T$ to the substring $S'$ under the \otom{} $\M{}$ ($\mu.o$) with $\mu.\dt=\delt{\M{}}=2$ deletions from the tree and $\mu.\ds=\dels{\M{}}=1$ deletions from the string.
    The start point of the derivation ($\mu.s$) is $3$. The end point of the derivation ($\mu.e$) is $8$. Notice that that $S_\M{}=F(T')$ and $T \succeq_2 T'$ which means that $S_\M{}\in C_{d_T}(T)$}
    \label{fig:derivation}
\end{figure}

\FloatBarrier

\begin{figure}[ht!]
	\centering
	\includegraphics[width=\linewidth]{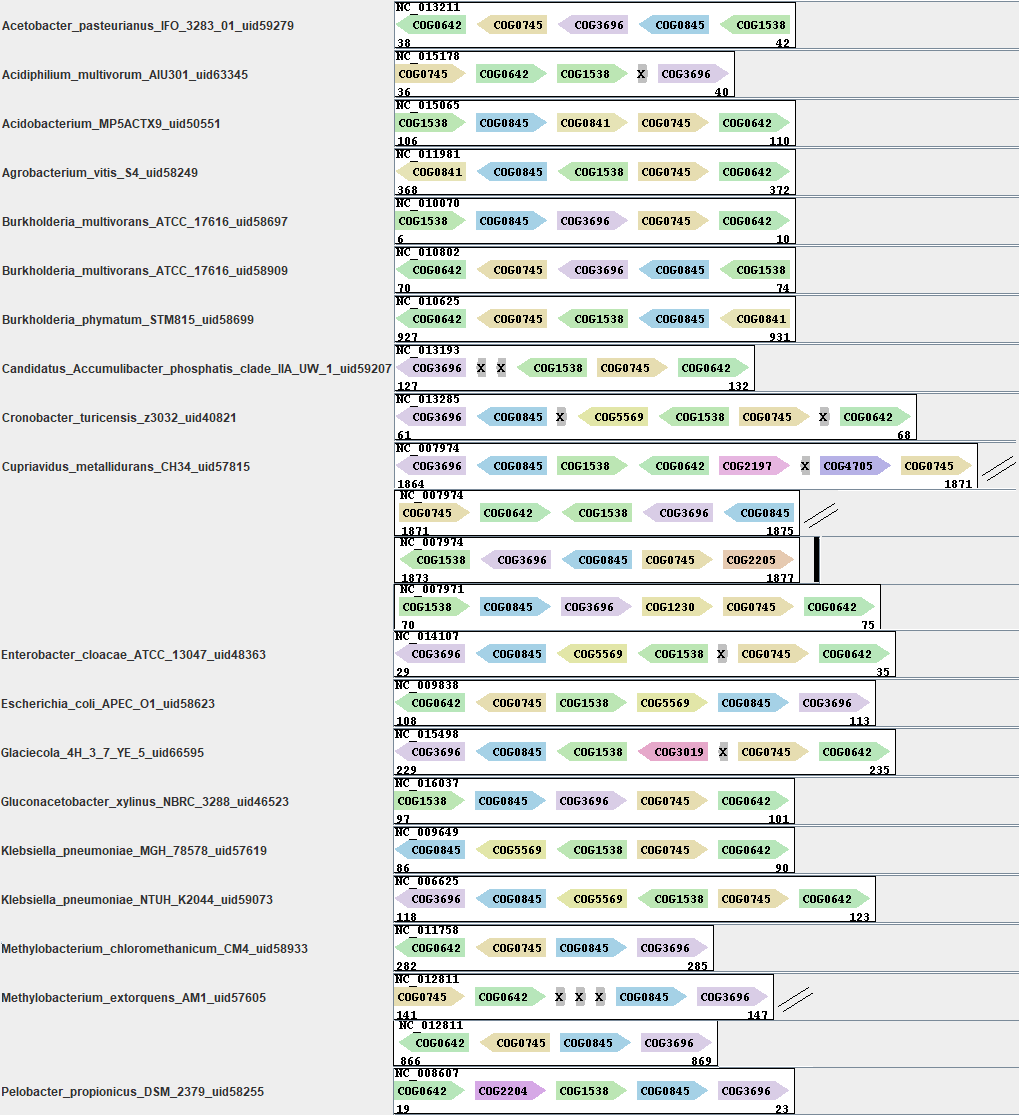}
	\caption{This figure in continued in the next page.}
\label{fig:pump_instances}
\end{figure}

\renewcommand{\thefigure}{S\arabic{figure} (Cont.)}
\addtocounter{figure}{-1}

\begin{figure}[ht]
	\centering
	\captionsetup[subfigure]{justification=centering}
	\begin{subfigure}[b]{1\textwidth}
	\centering
	\includegraphics[width=\linewidth]{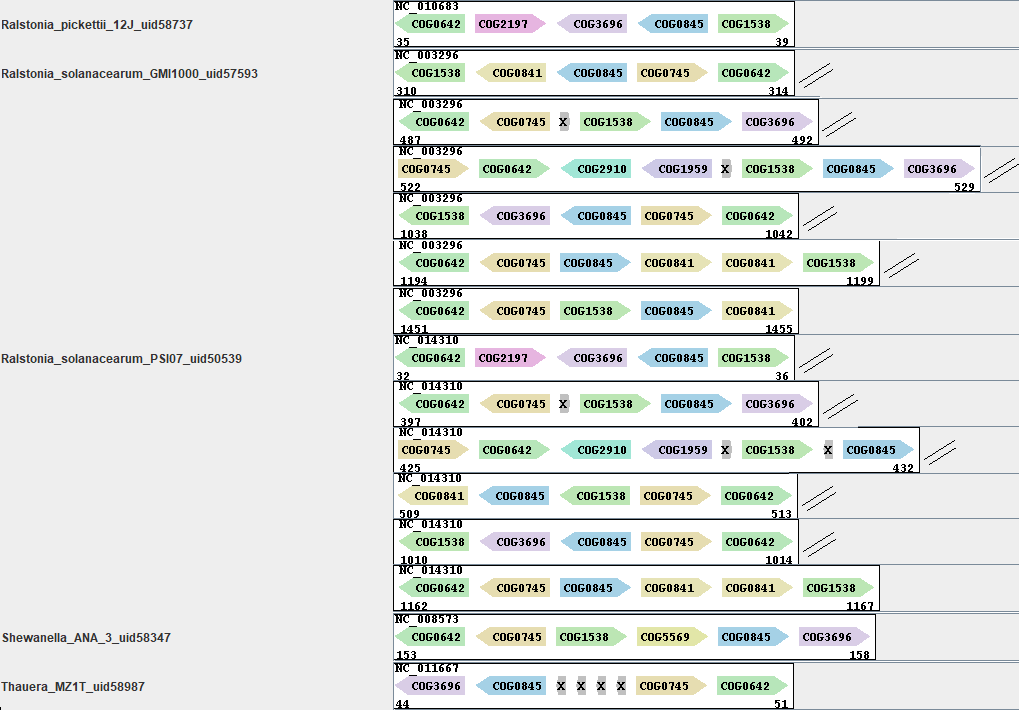}
    \caption{}
    \end{subfigure}
    \\
    \begin{subfigure}[b]{1\textwidth}
    \centering
	\includegraphics[scale=0.5]{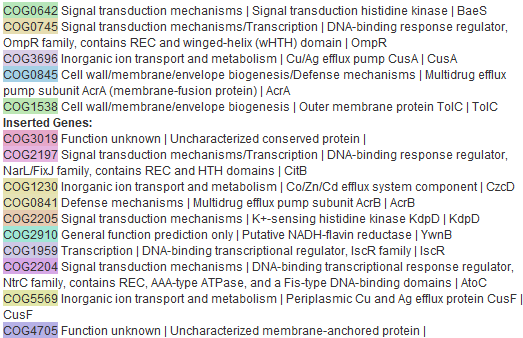}
	\caption{}
    \end{subfigure}

	\caption{ \textbf{(a)} The plasmid instances of the heavy metal efflux pump gene cluster discussed in \cref{sec:results_rnd}. The COGs of the query gene cluster are: COG0642, COG0745, COG3639, COG0845, COG1538. The instances were identified using \alg{} and displayed using the graphical interface of the tool CSBFinder-S \cite{csbfinder-s}. X indicates a gene with no COG annotation. The image was edited to display instances of the same genome in separate lines. \textbf{(b)} The functional description of the COGs shown in (a).  }
	
\label{fig:pump_instances2}
\end{figure}

\FloatBarrier

%% file: figures/derivation-string.tex
\begin{tikzpicture}
\node[] (s) at (-2.25, 0) {$S:$};
\node[] at (-1.5, 0) {$\sigma_1$};
\node[] at (-1, 0) {$\sigma_2$};
\node[] (s1) at (0, 0) {$\sigma_1$};
\node[] (s2) at (2, 0) {$\sigma_2$};
\node[] (s3) at (4, 0) {$\sigma_3$};
\node[] (s4) at (6, 0) {$\sigma_4$};
\node[] (s5) at (8, 0) {$\sigma_5$};
\node[] (s6) at (10, 0) {$\sigma_6$};
\node[] at (11, 0) {$\sigma_3$};
\draw[] (-0.5,0) -- (-0.5,0.5);
\draw[] (-0.5,0.5) -- (10.5,0.5);
\draw[] (10.5,0.5) -- (10.5,0);
\node[] (sTag) at (5, 1) {$S'$};
\node[below=of s] (m) {$\M{}:$};
\node[below=of s1] (m1) {$(x_3,\sigma_1(3))$};
\node[below=of s2] (m2) {$(\varepsilon,\sigma_2(4))$};
\node[below=of s3] (m3) {$(x_1,\sigma_3(5))$};
\node[below=of s4] (m4) {$(x_2,\sigma_4(6))$};
\node[below=of s5] (m5) {$(x_8,\sigma_5(7))$};
\node[below=of s6] (m6) {$(x_9,\sigma_6(8))$};
\draw[] (s1.south) -- (m1.north);
\draw[] (s2.south) -- (m2.north);
\draw[] (s3.south) -- (m3.north);
\draw[] (s4.south) -- (m4.north);
\draw[] (s5.south) -- (m5.north);
\draw[] (s6.south) -- (m6.north);
\node[below=of m] (t) {$S_{\M{}}:$};
\node[below=of m1] (t1) {$x_3$};
\node[below=of m3] (t3) {$x_1$};
\node[below=of m4] (t4) {$x_2$};
\node[below=of m5] (t5) {$x_8$};
\node[below=of m6] (t6) {$x_9$};
\draw[->] (m1.south) -- (t1.north);
\draw[->] (m3.south) -- (t3.north);
\draw[->] (m4.south) -- (t4.north);
\draw[->] (m5.south) -- (t5.north);
\draw[->] (m6.south) -- (t6.north);
\end{tikzpicture}

%% file: appendix-tables.tex
\section{Tables}

\begin{table}[h!]
\centering
\begin{tabular}{lllll}
\hline
{} &                                  PQ-tree$^1$ &   S-score & \# Genomes$^2$ &                                                                              Functional Category \\
\hline
1  &  [[0683 [[0411 0410] [0559 4177]]] 0583] &      22.5 &  5 (2) &  Amino acid transport \\
2  &  (1609 [1653 1175 0395] 3839) &      10.0 &  10 (2) &  Carbohydrate transport \\
3  &  [[1538 [3696 0845]] [0642 0745]] &       7.5 &  7 (1) &  Heavy metal efflux\\
4  &  [[2115 1070] [4213 [1129 4214]]] &       7.5 &  1 (1) &  Carbohydrate transport \\
5  &  [1960 [[2011 1135] [2141 1464]]] &       7.5 &  3 (1) &  Amino acid transport \\
6  &  [[0596 0599] [[3485 3485] 0015]] &       7.5 &  9 (1) &  Metabolism \\
7  &  [[[1129 1172 1172] 1879] 3254] &       7.5 &  6 (1) &  Carbohydrate transport \\
8  &  (1609 1869 [[1129 1172] 1879] 0524) &       7.5 &  1 (1) &  Carbohydrate transport \\
9  &  (0683 [0559 4177] [0411 0410] 0318) &       7.5 &  1 (1) &  Amino acid transport \\
10 &  (3839 0673 [[0395 1175] 1653]) &       5.0 &  10 (1) &  Carbohydrate transport \\
11 &  [0583 (0687 3842 [1176 1177])] &       5.0 &  9 (3) &  Amino acid transport \\
12 &  [1012 (0687 3842 [1176 1177])] &       5.0 &  8 (1) &  Amino acid transport \\
13 &  (0284 0461 [0540 1781] 0543 0044 0167) &       3.5 &  1 (1) &  Metabolism \\
14 &  ((2080 1319 1529) 1975 2068) &       3.3 &  6 (1) &  Energy production and conversion \\
15 &  [0044 [[0543 0167] 0284]] &       3.0 &  1 (1) &  Metabolism \\
16 &  [1802 [1638 [3090 1593]]] &       3.0 &  7 (1) &  Carbohydrate transport \\
17 &  [0410 [[4177 0559] 0683]] &       3.0 &  7 (3) &  Amino acid transport \\
18 &  [[4770 0511] [1984 2049]] &       3.0 &  4 (2) &  Metabolism \\
19 &  [[2875 [1010 2073]] 2243] &       3.0 &  9 (2) &  Metabolism \\
20 &  ([1175 0395] 1409 3839 1653) &       2.5 &  5 (2) &  Carbohydrate transport \\
21 &  [(2141 0431 0600 0715) 1116] &       2.5 &  2 (2) &  Inorganic ion transport \\
22 &  ([0601 1173] 0444 0444 0747) &       2.5 &  10 (1) &  Amino acid transport \\
23 &  [0583 (3842 1840 1178)] &       2.0 &  1 (1) &  Inorganic ion transport \\
24 &  (1464 2141 [1135 2011]) &       2.0 &  7 (3) &  Amino acid transport \\
25 &  ([2009 2142] 0479 1053) &       2.0 &  2 (1) &  Energy production and conversion \\
26 &  ([1622 0843] 0109 1845) &       2.0 &  1 (1) &  Energy production and conversion \\
27 &  (1024 1960 4770 4799) &       1.0 &  4 (1) &  Lipid transport \\
28 &  (1120 0609 0614 1629) &       1.0 &  4 (1) &  Inorganic ion transport \\
29 &  (0411 0559 4177 0683 0410 1022) &       1.0 &  3 (1) &  Amino acid transport \\
\hline

\end{tabular}
\caption{\pqt{s} for which tree-guided rearrangements were found in plasmids. $^1$Square brackets represent a Q-node; round brackets represent a P-node. Numbers indicate the respective COG IDs. $^2$This column indicates the number of genomes harboring plasmid instances of the respective \pqt. The number in brackets indicates the number of genomes harboring a tree-guided gene rearrangement of the corresponding gene cluster.}
\label{table:shuffling_full}
\end{table}

\FloatBarrier